\documentclass[envcountsect, 11pt, fleqn]{llncs}
\usepackage[left=2.5cm,right=2.5cm,top=3.0cm,bottom=3.0cm]{geometry}
\usepackage{amsmath}
\usepackage{amsfonts}
\usepackage{bm}
\usepackage{graphicx}
\usepackage{varwidth}
\usepackage{enumitem}
\usepackage{tikz}  %needed for \textcolor as well as tikz
\usetikzlibrary{arrows}

\usepackage[citecolor=blue, linkcolor=blue, colorlinks=true]{hyperref}

\renewcommand{\tabcolsep}{0.5em}
\newcommand\refTheorem[1]{\hyperref[#1]{Theorem~\ref*{#1}}}
\newcommand\refLemma[1]{\hyperref[#1]{Lemma~\ref*{#1}}}
\newcommand\refProposition[1]{\hyperref[#1]{Proposition~\ref*{#1}}}
\newcommand\refCorollary[1]{\hyperref[#1]{Corollary~\ref*{#1}}}
\newcommand\refFigure[1]{\hyperref[#1]{Figure~\ref*{#1}}}
\newcommand\refAlgorithm[1]{\hyperref[#1]{Algorithm~\ref*{#1}}}
\newcommand\refMechanism[1]{\hyperref[#1]{Mechanism~\ref*{#1}}}
\newcommand\refEquation[1]{\hyperref[#1]{(\ref*{#1})}}
\newcommand\refExample[1]{\hyperref[#1]{Example~\ref*{#1}}}
\newcommand\refSection[1]{\hyperref[#1]{Section~\ref*{#1}}}
\newcommand\refSubsection[1]{\hyperref[#1]{Subsection~\ref*{#1}}}
\newcommand\refDefinition[1]{\hyperref[#1]{Definition~\ref*{#1}}}
\newcommand\refRemark[1]{\hyperref[#1]{Remark~\ref*{#1}}}

\newcommand\supp[1]{\text{supp}\left(#1\right)}

\newcommand\trades{\leftrightarrow}

\begin{document}
\title{Fixed Price Approximability of the Optimal Gain From Trade}

\author{Riccardo Colini-Baldeschi \inst{1} \and Paul Goldberg \inst{2} \and Bart de Keijzer \inst{3}
\and \\ Stefano Leonardi  \inst{4}
%\and Tim Roughgarden \inst{5}
\and Stefano Turchetta \inst{5}}  % if Tim goes back on, ST's needs \inst{6} not \inst{5}

%If using runnningheads you can abbreviate the author name on even pages:
%\authorrunning{abbreviated author name}
%and you can change the author name in the table of contents
%\tocauthor{enhanced author name}

%For a single institute
\institute{LUISS Rome, \email{rcolini@luiss.it} \and University of Oxford, \email{paul.goldberg@cs.ox.ac.uk} \and Centrum Wiskunde \& Informatica (CWI), Amsterdam, \email{keijzer@cwi.nl}
\and Sapienza University of Rome, \email{leonardi@diag.uniroma1.it}
%\and Stanford University, \email{tim@cs.stanford.edu}
\and KPMG Italy, \email{stefano.turchetta@gmail.com}}

%If you're using runningheads you can add an abreviated title for the running head on odd pages using the following
%\titlerunning{abreviated title goes here}
%and an alternative title for the table of contents:
%\toctitle{table of contents title}

%\subtitle{Subtitle Goes Here}

%For a single author
%\author{Stefano Turchetta}

%For multiple authors:
%\author{Riccardo Colini-Baldeschi \and Paul Goldberg \and Bart de Keijzer \and Stefano Leonardi  \and Stefano Turchetta}

%If using runnningheads you can abbreviate the author name on even pages:
%\authorrunning{abbreviated author name}
%and you can change the author name in the table of contents
%\tocauthor{enhanced author name}

%For a single institute
%\institute{LUISS Universit\`{a} Guido Carli \email{rcolini@luiss.it} \and University of Oxford \email{paul.goldberg@cs.ox.ac.uk}, \email{stefano.turchetta@gmail.com} \and Centrum Wiskunde \& Informatica (CWI) \\ \email{keijzer@cwi.nl} \and Sapienza University of Rome \email{leonardi@dis.uniroma1.it}}

% If authors are from different institutes 
%\institute{University of Oxford \email{paul.goldberg@cs.ox.ac.uk} \and Technische Universit\"at M\"unchen \email{stefano.turchetta@gmail.com}}

%to remove your email just remove '\email{email address}'
% you can also remove the thanks footnote by removing '\thanks{Thank you to...}'

\maketitle

\begin{abstract}
\emph{Bilateral trade} is a fundamental economic scenario comprising a strategically acting buyer and seller (holding an item), each holding valuations for the item, drawn from publicly known distributions. A mechanism is supposed to facilitate trade between these agents, if such trade is beneficial. It was recently shown that the only mechanisms that are simultaneously dominant strategy incentive compatible, strongly budget balanced, and ex-post individually rational, are \emph{fixed price} mechanisms, i.e., mechanisms that are parametrised by a price $p$, and trade occurs if and only if the valuation of the buyer is at least $p$ and the valuation of the seller is at most $p$.

The {\em gain from trade} is the increase in welfare that results from applying a mechanism; here we study the gain from trade achievable by fixed price mechanisms. We explore this question for both the bilateral trade setting, and a \emph{double auction} setting where there are multiple buyers and sellers. We first identify a fixed price mechanism that achieves a gain from trade of at least $2/r$ times the optimum, where $r$ is the probability that the seller's valuation does not exceed the buyer's valuation. This extends a previous result by McAfee. Subsequently, we improve this approximation factor in an asymptotic sense, by showing that a more sophisticated rule for setting the fixed price results in an expected gain from trade within a factor $O(\log(1/r))$ of the optimal gain from trade. This is asymptotically the best approximation factor possible. 

Lastly, we extend our study of fixed price mechanisms to the double auction setting defined by a set of multiple i.i.d. unit demand buyers, and i.i.d. unit supply sellers. We present a fixed price mechanism that achieves a gain from trade that achieves for all $\epsilon > 0$ a gain from trade of at least $(1-\epsilon)$ times the expected optimal gain from trade with probability $1 - 2/e^{\#T \epsilon^2 /2}$, where $\#T$ is the expected number of trades resulting from the double auction. 
This can be interpreted as a ``large market'' result: Full efficiency is achieved in the limit, as the market gets thicker.
\end{abstract}

%\newpage

\section{Introduction}
\emph{Bilateral trade} is a fundamental economic scenario comprising a buyer and a seller.
The seller holds one item, and can possibly trade this item with the buyer for some price.
The buyer and the seller each have a (non-negative real-valued) valuation for the item that is up for trade.
The buyer's valuation is only known by the buyer and the seller's valuation is only known by the seller. The buyer and seller both want to maximise their utility, which is assumed to be \emph{quasi-linear}, i.e., of the form $x \cdot v - p$, where $x$ is a $0/1$-variable that is set to $1$ if and only if the agent holds the item, $v$ is the agent's value for the item, and $p$ is the price paid/received by the agent. In the buyer's case $p$ is non-negative and represents how much the buyer has to pay. In the seller's case, the price $p$ is non-positive because the seller receives money to transfer her item.

The main problem studied for this bilateral trade setting is one in mechanism design: which mechanism maximises the \emph{social welfare} (i.e., total utility of both players)? A direct revelation mechanism for this setting solicits the valuations of the buyer and the seller. Subsequently it determines whether the buyer and the seller should trade and which prices  they have to pay or receive. We would like any mechanism to satisfy the following properties:
\begin{itemize}
 \item Dominant strategy incentive compatibility (DSIC): It should be a dominant strategy for the buyer and seller to submit their true valuations to the mechanism.
 \item Ex-post individual rationality (ex-post IR): Neither agent should end up with a negative utility if the agent's true valuation is submitted to the mechanism.
 \item Strong budget balance (SBB): The price paid by the buyer is equal to the price received by the seller, i.e., the mechanism does not extract money from the market, nor does it inject money into the market.
\end{itemize}

While the valuations of the buyer and seller are known by the buyer and seller only, it is assumed that there is still distributional public knowledge about their valuations. More precisely, it is assumed that there are two publicly known distributions from which the buyer and seller independently draw their valuations. The mechanism may thus make use of these distributions in order to determine the outcome.

Ideally, we would want the mechanism to have the seller trade with the buyer whenever the buyer's valuation exceeds the seller's valuation. The expected total utility that would result from trading as such is referred to as the \emph{optimal} social welfare. Unfortunately the optimal social welfare is not achievable, as shown by Myerson and Satterthwaite~\cite{myersonsatterthwaite}: No bilateral trade mechanism is simultaneously DSIC, IR, {\em weakly} budged balanced, and social-welfare optimizing. Weak budget balance (WBB) is less restrictive than strong budget balance, as WBB only requires that no money be injected into the market, while the mechanism is allowed to extract money from the market.

For the classic bilateral trade setting, it was recently shown \cite{doubleauctions} that the only direct revelation mechanisms that are simultaneously incentive compatible, strongly budget balanced, and ex-post individually rational, are \emph{fixed price} mechanisms, i.e., mechanisms that are parametrised by a price $p$, and trade occurs if and only if the valuation of the buyer is at least $p$ and the valuation of the seller is at most $p$. %The same paper also established the currently best known upper and lower bounds on the factor by which such mechanisms can approximate the optimal social welfare. The bounds are $1{\cdotp}33$ and $1{\cdotp}92$ respectively. 

An alternative ---and more challenging to approximate--- objective to the social welfare is the \emph{gain from trade}, which measures the expected increase in total utility that is achievable by applying the mechanism, with respect to the initial allocation.
For example, if a seller holds an item that she values $\$4$ and a buyer values the same item $\$10$, whenever a fix price mechanism sets a price $4 \leq p \leq 10$, the buyer and the seller trade producing a gain from trade of $\$6$. Whenever the price $p$ is set lower than $\$4$ or greater than $\$10$ no trade occur, and the gain from trade is $0$.

McAfee~\cite{mcafee} has shown that if the median of the distribution of the seller's valuation is less than the median of the distribution of the buyer's valuation, then there is a fixed price mechanism for which the expected gain from trade is at least half of the optimal gain from trade. In fact, it was shown for this special case that by setting the fixed price anywhere in between the two medians, half of the optimal gain from trade is guaranteed. We extend this result by showing that the optimal gain from trade is at least $2/r$ times the gain from trade achievable by a fixed price mechanism, where $r$ is the probability that the seller's valuation does not exceed the buyer's valuation (which is the condition under which a gain from trade is possible in the first place). 

Subsequently, we improve this approximation factor in an asymptotic sense, by showing that a more sophisticated rule for setting the fixed price results in an expected gain from trade within a factor $O(\log(1/r))$ of the optimal gain from trade. This is asymptotically the best approximation factor possible, which is shown by an appropriate example of a bilateral trade setting for which every fixed price achieves an expected gain from trade of $\Omega(\log(1/r))$ times the expected gain from trade.

It follows from our results that our mechanisms cannot approximate the gain from trade if the probability of trading is small. Indeed, we  prove a general negative result showing that the ratio between the gain from trade of a DSIC mechanism and the optimal gain from trade can be arbitrarily small as the support of the distribution grows.  A similar result has been proved independently in \cite{BD16}.

We finally extend our study to the \emph{double auction} setting, where there are multiple buyers and sellers, each seller holding one item and each buyer having a demand for obtaining at most one item. The valuations of the $n$ buyers are independently drawn from a common probability distribution, and the same holds for the $m$ sellers, although the probability distribution of the sellers may be distinct from that of the buyers.

\subsection{Our Results}
The first results presented in this paper concern the bilateral trade problem. It is known that if a mechanism has to satisfy IR, DSIC, and SBB, then it must be a fixed price mechanism \cite{doubleauctions}, i.e., the mechanism fixes a price $p$ and posts it to the buyer and the seller. We want to understand how this price $p$ has to be chosen.

McAfee's result of \cite{mcafee} states that in case the seller's median is less than the buyer's median, then setting the price in between the medians of the buyer and the seller results in a $2$-approximation to the optimal gain from trade.
Our first result is a strict generalization of \cite{mcafee} where the approximation to the optimal gain from trade is given as a function of the probability that a trade is efficient, in other words: the probability that the valuation drawn from the buyer is greater than the valuation drawn from the seller. This parameter is referred to as $r = \mathbf{Pr}_{v \sim f, w \sim g}[v \geq w]$, where $f$ is the buyer's distribution and $g$ is the seller's distribution.

In particular, we show that setting the price $p$ such that $\mathbf{Pr}_{v \sim f}[v \geq p] = \mathbf{Pr}_{w \sim g}[w \leq p]$ results in a $r/2$-approximation to the optimal gain from trade.

Then, we show how that it is possible to improve the approximation factor of $2/r$ considerably in an asymptotic sense: We prove that by using a more complex rule for determining the fixed price $p$, the optimal gain from trade is at most a factor of $O(\log(1/r))$ times the gain from trade when trading at price $p$. When $r$ is small, this results in a big improvement when compared to the approximation factor that we established in the previous section.
Our mechanism works by showing that we can decompose ``roughly'' the entire probability space into at most $\log(1/r) + 1$ such events, so that choosing the best fixed price corresponding to each of these events results in a gain from trade that is an $O(\log(1/r))$-approximation to the optimal gain from trade. 
%More precisely, we show that there are two sets of roughly $\log(1/r) + 1$ such events, and we prove that in case one of these sets does not cover a fraction of the probability space that accounts for at least $1/2$ of the optimal gain from trade, then the other set of events does. 

Finally, we want to consider the double auction setting. In this setting, we extend the definition of a fixed price mechanism in a natural way: the mechanism computes a single price $p$, buyers with a valuation greater than $p$ and sellers with a valuation lower than $p$ will be allowed to trade. If the sets of allowed buyers and allowed sellers have different cardinalities, agents will be removed from the biggest set uniformly at random so that the cardinality of the two sets will be equal.%such a mechanism is parametrized by a single price $p$ (just as for the bilateral trade setting) and the mechanism lets a maximal uniform random subset of \emph{feasible pairs} trade, where a feasible pair consists of a buyer whose valuation exceeds $p$, and a seller whose valuation is at most $p$.

The fixed price mechanism that we propose for the double auction setting achieves a gain from trade that is a $1-\epsilon$ approximation to the optimal gain from trade with probability $(1 - 1/e^{\epsilon^2\#T/2})$ where $\#T$ is the expected number of trades of the mechanism. This implies that if the double auction instance is such that a relatively small expected number of trades can happen at this price, then a reasonably good approximation factor is achieved by our mechanism (see Section \ref{sec:double-auctions} for a detailed discussion). One may also interpret our result as a ``large market'' result: the approximation factor approaches $1$ as we let the number of buyers and sellers in the market grow proportionally, since in that case the number of trades grows arbitrarily large. This is, to the best of our knowledge, the first fixed price mechanism that is DSIC, SBB, and ex-post IR, and achieves a near-optimal gain from trade under mild conditions on the size of the market.

\subsection{Related Literature}
The impossibility result of \cite{myersonsatterthwaite} proved that no two-sided mechanism can be simultaneously BIC, IR, WBB, and optimise the social welfare even in the simple bilateral trade setting. Thus, many subsequent works studied how it is possible to relax some of the constraints to achieve positive results in the context of maximise the social welfare or the gain from trade.

In \cite{Blumrosen2016}, a BIC mechanism is devised that approximates the expected gain from trade in bilateral trade up to a 
factor of $1/e$ when the buyer's distribution function satisfies a property known as the monotone hazard rate condition. The mechanism of \cite{Blumrosen2016} is not DSIC since it achieves this approximation factor from a Bayes-Nash equilibrium by using the valuation of the seller in the price offered to the buyer.  It is also shown in the same work that no BIC mechanism can achieve an approximation bound better than $2/e$. 
Mechanisms that are DSIC/BIC, IR, and SBB have been given for bilateral trade in \cite{bd14}. In addition to this, the authors proposed a WBB mechanism for a general class of markets known as combinatorial exchange markets. 

Mechanisms for double auctions with near-optimal gain from trade have been previously proposed for the prior-free setting. McAfee~\cite{mcafee92} has shown a WBB, DSIC, IR mechanism which achieves a $1-1/k$-approximation to the optimal gain from trade if the number of trades under the optimum allocation is $k$. This rate of convergence requires the prior distributions of traders to be bounded above zero, over an interval $[0,1]$, but the mechanism is not a function of the priors. More recently, Segal-Halevi et al.~\cite{HHA16} devised a SBB mechanism with the same performance guarantee. %TODO: check and compare better 
The mechanisms of \cite{mcafee92,HHA16} are direct revelation mechanisms where the price depends crucially on the reported valuations.
In contrast, the goal in our present paper is to find out how much gain from trade can be generated by means of setting a single fixed price, independent of the valuations of the players, at which all agents trade. Such mechanisms have the advantage that they are conceptually simpler and have a pricing scheme that is extremely easy to understand. They also can be implemented as posted price mechanisms where buyers and sellers merely have to indicate whether they would be willing to trade at the given price. Thus, they require a minimum amount of valuation revelation.
 
In \cite{DBLP:journals/corr/Segal-HaleviHA16}, the authors present a mechanism that combines random sampling and random serial dictatorship techniques which is IR, SBB and DSIC, and asymptotically approaches the optimum gain from trade. Recently, \cite{DBLP:conf/sigecom/BrustleCWZ17} provides an IR, SBB, and BIC mechanism that achieves a constant approximation to the best gain from trade achievable among the IR, WBB, and BIC mechanisms, which is an alternative (more permissive) benchmark. Two recent papers by Feldman and Gonen \cite{feld1,feld2} study a multi-unit variant of double auctions for online advertising purposes. They design in \cite{feld1} IR, WBB, and DSIC mechanisms that well-approximate the gain from trade under certain technical conditions, as a function of the number of trades under the optimum allocation. In \cite{feld2} they develop further mechanisms for an online variant of this setting, where additional sellers and buyers may enter the market over time.

Our paper contributes to a recent line of computational work on {\em market intermediation}, in which a mechanism is assumed to interact
with both buyers and sellers of some good, or goods. The objective is to maximise (or, approximately maximise) either revenue for the mechanism, or welfare of the buyers/sellers. Deng et al.~\cite{DGTZ14} study revenue maximisation, in a setting of multiple buyers and sellers, with uncorrelated priors and a single type of item being traded. The same objective was studied by \cite{dghk02} yet in the \emph{prior-free} model. Gerstgrasser et al.~\cite{GGK16} also study the objective of maximising expected revenue, in a setting where there is a small number of buyers and seller, who have a prior distribution whose support size represents the complexity of instances of the problem. In \cite{GGK16}, this distribution is otherwise unrestricted, and in particular may be correlated.
Giannakopoulos et al.~\cite{gkl17} study a similar double auction setting to the one studied here in Section~\ref{sec:double-auctions}: there are multiple buyers and sellers and one kind of item, with unit supply and demand. Buyers have a common prior distribution on their valuations, as do sellers. \cite{gkl17} study market intermediation from the perspective of welfare maximisation, and revenue maximisation. Colini-Baldeschi et al.~\cite{CGKLRT17} also study market intermediation, in the context of buyers and seller of a collection of heterogenous items, aiming to maximise welfare and achieve a strong notion of budget balance.

In the context of social welfare in the economics literature it is often studied a setting in which the valuation of the sellers and the buyers are independently drawn from identical regular distributions, while satisfying IR and WBB. And in this setting the goal is to find an approximation of the optimal social welfare as a function of the number of agents. \cite{gs89} showed that duplicating the number of agents by $\tau$ results in a market where the optimal IR, IC, WBB mechanism's expected social welfare approximation factor approaches $1$ at a rate of $O(\log{\tau} / \tau^2)$. \cite{rsw94} and \cite{sw02} investigated a family of non-IC double auctions, and study the inefficiency and the extent to which agents misreport their valuations in these double auctions.

\section{Preliminaries}\label{sec:prelims}
As a general convention, we use $[a]$ to denote the set $\{1, \ldots, a\}$. We will use $\mathbf{1}(X)$ to denote the indicator function that maps to $1$ if and only if event/fact $X$ holds. 

\paragraph{Double Auction Setting.}
In a double auction setting there are $n$ buyers and $m$ sellers. Initially, each seller $j \in [m]$ holds one item and has a valuation $w_{j}$ for it. The sellers are not interested in possess more than one item.
Each buyer $i \in [n]$ is interested in obtaining no more than one item and has a valuation $v_{i}$ for it. Moreover, they are indifferent among the different items.

The valuations of the buyers and the sellers are private knowledge, but they are independently drawn from publicly known distributions $f$ and $g$, where $f$ is the probability distribution for the valuation of a buyer and $g$ is the probability distribution for the valuation of a seller.
We treat $f$ and $g$ as probability density functions. All the buyers share the same distribution $f$ and all the sellers share the same mass probability distribution $g$, but $f$ and $g$ may be distinct.
Moreover, let $G$ be the corresponding cumulative distribution functions of $g$ and let $\bar{F}$ be the corresponding complementary cumulative distribution function (or survival function) of $f$.

Given a double auction setting $(n, m, f, g)$, our goal is to redistribute the items from the sellers to the buyers. An \emph{allocation} for a double auction setting $(n, m, f, g)$ is a pair of vectors $(\bm{X},\bm{Y}) = ((X_1, \ldots, X_n),(Y_1,\ldots,Y_m))$ such that all the elements $X_1, \ldots, X_n, Y_1, \ldots, Y_m \in \{0, 1\}$, and $\sum_{i \in [n]} X_{i} + \sum_{j \in [m]} Y_{j} = m$. The set $\mathcal{A}$ represents the set of all allocations for the double auction setting.

The redistribution of the items from sellers to buyers is done by running a \emph{mechanism} $\mathbb{M}$.
A mechanism receives input from the agents, and outputs an allocation $(\bm{X},\bm{Y})$ and a price $p$.
The allocation $(\bm{X},\bm{Y})$ and the price $p$  represents the \emph{outcome} of the mechanism $\mathbb{M}$.
Thus, an outcome is a tuple $(\bm{X},\bm{Y}, p)$. The price $p$ represents how much a buyer has to pay to obtain an item and how much a seller has to receive to sell her item.\footnote{More generally, we may define the notion of a mechanism such that more complex pricing schemes are possible, but our definition suffices for the mechanisms that we will define later in this paper.}

Agents are assumed to be utility maximisers. The \emph{utility} is defined as the valuation for the items that they possess with respect to the allocation vector, minus the payment charged by the mechanism. Specifically, the utility of a buyer $i$ will be $u_{i}^B(\bm{X},\bm{Y}, p) = (v_{i} - p)\cdot X_{i}$. Similarly, the utility of a seller $j$ will be $u_{j}^S(\bm{X},\bm{Y}, p) = w_{j} Y_{j} + p \cdot (1-Y_j)$. 

Furthermore, agents are assumed to be fully rational, so that they will strategically interact with the mechanism to achieve their goal of maximising utility. Our goal is to design a mechanism that is DSIC, IR, SBB (as defined in the introduction) such that the the \textit{gain from trade} is high. For an outcome $(\bm{X}, \bm{Y}, p)$, the gain from trade $\textrm{GFT}(\bm{X},\bm{Y}, p)$ is defined as
\begin{equation*}
\textrm{GFT}(\bm{X},\bm{Y}, p) = \sum_{i=1}^n v_i X_i + \sum_{j=1}^m w_j (Y_j - 1) 
\end{equation*}

%Moreover, the mechanism has to satisfy the following additional properties.

%\begin{itemize}

%\item \textbf{Dominant strategy incentive compatibility (DSIC)}: It is a dominant strategy for every agent to report her true valuation sincerely. I.e., for every agent $i$ and for every vector of valuations of all other players, it is impossible for agent $i$ to increase her expected utility by misreporting her valuation. A weaker variant of this is \emph{Bayesian} incentive compatibility (BIC), which only requires that it is a Bayes-Nash equilibrium when all agents report their valuations truthfully.

%\item \textbf{Ex-post individual rationality (ex-post IR):} It is not harmful for any agent to participate in the mechanism, i.e., there is guaranteed to be a strategy for an agent that yields the agent a utility that is not less than her initial utility.

%\item \textbf{Strong Budget Balance (SBB):} The sum of all agents' payments output by the mechanism \emph{is equal to} zero. Conceptually, this means that no money ends up at an external party, and no external party needs to subsidise the mechanism.

%\end{itemize}

For a double auction setting $(n, m, f, g)$, the \textit{expected optimal gain from trade} is defined as
\begin{equation*}
\text{OPT}_{n,m,f,g} = \mathbf{E}_{v \sim f^n, w \sim g^m}\left[\max\left\{\sum_{i=1}^n v_i X_i + \sum_{i=1}^m w_j (Y_j - 1)\ \middle|\ (\bm{X}, \bm{Y}) \in \mathcal{A} \right\}\right].
\end{equation*}
We will sometimes omit the subscript, as in those cases the instance being discussed will be clear from context.

We say that a mechanism $\mathbb{M}$ \emph{$\alpha$-approximates the optimal gain from trade} for some $\alpha > 1$ if and only if $\text{OPT} \leq \alpha \mathbf{E}[\text{GFT}(\bm{X},\bm{Y}, p)]$, where $(\bm{X},\bm{Y},p)$ is the random allocation that the mechanism generates, when valuations $v$ and $w$ are drawn from $f^n$ and $g^m$ respectively. Our goal is to find a DSIC, ex-post IR, and SBB mechanism that $\alpha$-approximates the optimal gain from trade for a low $\alpha$.

\paragraph{Bilateral Trade Setting.}
The bilateral trade setting is a special case of the double auction setting where there is only one unit-demand buyer and one unit-supply seller. Thus, we can represent a bilateral trade setting as a pair of valuation distribution function, one for the buyer $f$ and one for the seller $g$, i.e., $(f,g)$.
It is known that if a mechanism has to satisfy IR, DSIC, and SBB, then it must be a fixed price mechanism \cite{doubleauctions}, i.e., the mechanism fixes a price $p$ a priori, and trade happens if and only if both the buyer's valuation is at least $p$ and the seller's valuation is at most $p$.

For a bilateral trade instance, the gain from trade of a fixed price mechanism with fixed price $p$ will be denoted by $\text{GFT}_{f,g}(p)$. That is,
\begin{equation*}
\text{GFT}_{f,g}(p) = \mathbf{E}_{v \sim f, w \sim g}[\max\{0,v-w\}\mathbf{1}(w \leq p \leq v)].
\end{equation*}
Moreover, note that for the bilateral trade setting we can express $\text{OPT}_{f,g}$ as $\mathbf{E}_{v \sim f, w\sim g}[\max\{0,v-w\}]$.

For the bilateral trade setting, the goal of this paper is to study how to set the price $p$ such that the gain from trade achieved by the fixed price mechanism with price $p$ is as close as possible to the optimal gain from trade.
We will design fixed price mechanisms where the ratio between $\text{OPT}_{f,g}$ and $\text{GFT}_{f,g}(p)$ is a function of the probability that the buyer has a value greater than the seller, i.e., the provability that a trade is efficient.
This probability will be represented by the parameter $r$. Thus,
\begin{equation*}
r = \mathbf{Pr}_{v \sim f, w \sim g}[v \geq w].
\end{equation*}

Due to space constraints, the proofs of various theorems and lemmas have been deferred to the appendix.

\section{An $O(1/r)$-Approximation Mechanism for Bilateral Trade}
In the bilateral trade setting there is only one unit-demand buyer and one unit-supply seller. It can be proven that if a mechanism has to satisfy IR, DSIC, and SBB, then it must be a fixed price mechanism \cite{doubleauctions}, i.e., the mechanism fixes a price $p$ and posts it to the buyer and the seller. Trade happens if and only if both the buyer's valuation is at least $p$ and the seller's valuation is at most $p$.

We will show in this section that there exists a fixed price mechanism that achieves an expected gain from trade that is at least $r/2$ times the expected optimal gain from trade. In the fixed price mechanism that we propose for this, the fixed price $p$ is set such that $\mathbf{Pr}_{v \sim f}[v \leq p] = \mathbf{Pr}_{w \sim g}[w \leq p]$. The main theorem that we prove is thus as follows.
\begin{theorem}\label{mainthm}
Let $(f,g)$ be a bilateral trade instance, let $p \in \mathbb{R}_{\geq 0}$ be any fixed price, and let $q$ be the minimum of $\mathbf{Pr}_{v \sim f}[v \geq p]$ and $\mathbf{Pr}_{w \sim g}[w \leq p]$. Then,
\begin{equation}\label{eq:1}
\frac{1}{q} \text{GFT}_{f,g}(p) \geq \text{OPT}_{f,g} .
\end{equation}
Moreover, if $p$ is chosen such that $q$ is maximised (i.e., $p$ is such that $\mathbf{Pr}_{w \sim g}[w \leq p] = \mathbf{Pr}_{v \sim f}[v \geq p]$), it holds that
\begin{equation}\label{eq:2}
\frac{2}{r} \text{GFT}_{f,g}(p) \geq \text{OPT}_{f,g}.
\end{equation}
\end{theorem}
Note that this theorem strictly generalises McAfee's result of \cite{mcafee}, which states that in case the seller's median is less than the buyer's median, then setting the price in between the medians of the buyer and the seller results in a $2$-approximation to the optimal gain from trade: If we take $p$ to be any price in between the median of the seller and the buyer, then $q$ is at least $1/2$, and (\ref{eq:1}) then states that the gain from trade at fixed price $p$ is at least half the optimal gain from trade.

\section{Improving the Asymptotic Dependence on $r$}
In this section, we show how it is possible to improve the approximation factor implementing a more involved rule to determine the fixed price $p$. When the trading price $p$ will be set with the new rule the approximation factor will improve from $2/r$ to $O(\log(1/r))$. Notice that when $r$ is small, this is a big improvement with respect to the approximation shown in the previous section.
All logarithms used in this section are to base 2. 

Let us first give a high level description of how we determine the fixed price of the mechanism. Let us consider any two points $z$ and $z'$ such that $\mathbf{Pr}[v \geq z] = 2 \mathbf{Pr}[v \geq z']$. Let $E$ be the event that the buyer's valuation exceeds $z$, and that the sellers valuation lies in between $z$ and $z'$. Let $\overline{F}_{E}$ be the complementary cumulative distribution function of the buyer conditioned on $E$ and let $G_{E}$ be the cumulative distribution function of the seller conditioned on $E$. We now see that on the interval $[z,z']$, the function $\overline{F}_{E}$ decreases from $1$ to $1/2$ and the function $G_{E}$ increases from $0$ to $1$. Thus, the functions cross each other in $[z,z']$ at a value of at least $1/2$, which means that the median of the buyer exceeds the median of the seller when conditioning on $E$. Using Theorem \ref{mainthm}, we thus obtain that when conditioning on $E$ there exists a fixed price that achieves a $2$-approximation to the optimal gain from trade.

Our mechanism works by showing that we can decompose ``roughly'' the entire probability space into at most $\log(1/r) + 1$ such events, so that choosing the best fixed price corresponding to each of these events results into an $O(\log(1/r))$ approximation to the optimal gain from trade. More precisely, we show that there are two sets of roughly $\log(1/r) + 1$ such events, and we prove that in case one of these sets does not cover a fraction of the probability space that accounts for at least $1/2$ of the optimal gain from trade, then the other set of events does. To determine the desired fixed price, we can thus
\begin{enumerate}
\item first determine which of the two event sets ``covers'' a large part of the optimal gain from trade, 
\item and subsequently select the best fixed price among the $\log(1/r) + 1$ prices corresponding to the event set.
\end{enumerate}
The two event sets have the following properties: one of them excludes the part of the probability space where the buyer's complementary CDF is below the threshold $r/2$. The other one switches the roles of the seller and buyer, and excludes the part of the probability space where the seller's CDF is below a the threshold $r/2$. From this property of the event sets (i.e., having these particular thresholds on the tails of the two distributions), we are able to show that one of the event sets covers a large part of the optimal gain from trade.
We now proceed by making these ideas precise.

We first describe how we determine the price, which we denote by $p^*$, for a given instance $(f,g)$. In contrast with the last section, we assume (for convenience of exposition) without loss of generality that $f$ and $g$ are continuous distributions without point masses, where we treat $f$ and $g$ as probability density functions, and we let $F$ and $G$ be the corresponding cumulative distribution functions. We write $\overline{F}$ to denote the buyer's complementary cumulative distribution function $1 - F$. Let $r$ be the probability $\mathbf{Pr}_{v \sim f, w \sim g}[v \geq w]$ of a trade being possible (as before). Let $x$ be the value such that $F(x) = r/2$ and let $y$ be the value such that $\overline{G}(y) = r/2$.
We distinguish between two cases.
\begin{itemize}
 \item If $\mathbf{E}_{v \sim f, w \sim g}[(v-w)\mathbf{1}(w \leq v \wedge w > y)] 
 %\leq \mathbf{E}_{v \sim f, w \sim g}[(v-w)\mathbf{1}(w \leq v)]/2 
 \geq \text{OPT}_{f,g}/2$, then let $p^*$ be the price that achieves the maximum gain from trade among the prices $p_1, \ldots, p_{\lceil\log(2/r)\rceil}$, where for $i \in [\lceil\log(2/r)\rceil]$, price $p_i$ is such that
\begin{equation*}
\mathbf{Pr}_{w \sim g}\left[w \leq p_i\ \middle|\ \overline{F}^{-1}\left(\frac{1}{2^{i-1}}\right) \leq w \leq \overline{F}^{-1}\left(\frac{1}{2^i}\right)\right] = \mathbf{Pr}_{v \sim f}\left[v > p_i\ \middle|\ \overline{F}^{-1}\left(\frac{1}{2^{i-1}}\right) \leq v\right].
\end{equation*}
 \item Otherwise, let $p^*$ be the price that achieves the maximum gain from trade among the prices $p_1', \ldots, p_{\lceil\log(2/r)\rceil}'$, where for $i \in [\lceil\log(2/r)\rceil]$, price $p_i'$ is such that 
\begin{equation*}
\mathbf{Pr}_{v \sim f}\left[v > p_i\ \middle|\ G^{-1}\left(\frac{1}{2^{i}}\right) \leq v \leq G^{-1}\left(\frac{1}{2^{i-1}}\right)\right] = \mathbf{Pr}_{w \sim g}\left[w \leq p_i\ \middle|\  G^{-1}\left(\frac{1}{2^{i}}\right) \leq w\right],
\end{equation*}
where we define $G^{-1}(1) = \infty$ if there exists no point $t \in \mathbb{R}_{\geq 0}$ such that $G(t) = 1$.
\end{itemize}
This completes the definition of the fixed price $p^*$.

First we prove that if the first of the two cases does not apply (i.e., if the inequality $\mathbf{E}_{v \sim f, w \sim g}[(v-w)\mathbf{1}(w \leq v \wedge w > y)] \leq \text{OPT}_{f,g}/2$ does not hold), then the symmetric inequality $\mathbf{E}_{v \sim f, w \sim g}[(v-w)\mathbf{1}(w \leq v \wedge w < x)] \leq \text{OPT}_{f,g}/2$ holds for the second case.

\begin{lemma}\label{caseslemma}
If $\mathbf{E}_{v \sim f, w \sim g}[(v-w)\mathbf{1}(w \leq v \wedge w > y)] > \text{OPT}_{f,g}/2$, then $\mathbf{E}_{w \sim f, v \sim g}[(v-w)\mathbf{1}(w \leq v \wedge v < x)] \leq \text{OPT}_{f,g}/2$.
\end{lemma}
\begin{proof}

It suffices to prove that the events $w \leq v \wedge w > y$ and $w \leq v \wedge v < x$ are disjoint, since then the value $\text{OPT}_{f,g}$ can be decomposed into three terms:
\begin{eqnarray*}
\text{OPT}_{f,g} & = & \mathbf{E}_{v \sim f, w \sim g}[(v-w)\mathbf{1}(w \leq v)\mathbf{1}(w > y)] + \mathbf{E}_{v \sim f, w \sim g}[(v-w)\mathbf{1}(w \leq v)\mathbf{1}(w < x)] \\
& & \qquad + \mathbf{E}_{v \sim f, w \sim g}[(v-w)\mathbf{1}(w \leq v)\mathbf{1}(v \leq y \wedge w \geq x)],
\end{eqnarray*}
from which it follows that at least one of the first two terms does not exceed $\text{OPT}_{f,g}/2$, which proves the claim.

In order to show that $w \leq v \wedge w > y$ is disjoint from $w \leq v \wedge v < x$, in turn it suffices to prove that $y \geq x$. To see this, suppose for contradiction that $y < x$. Then we can derive
\begin{eqnarray*}
\mathbf{Pr}_{v \sim f, w \sim g}[w \leq v] & \leq & \mathbf{Pr}_{v \sim f, w \sim g}[w \leq x \cap w \leq v] + \mathbf{Pr}_{v \sim f, w \sim g}[w > x \cap w \leq v]  \\
& \leq & \mathbf{Pr}_{w \sim g}[w \leq x] + \mathbf{Pr}_{v \sim f}[v > x] \\
& < & \mathbf{Pr}_{w \sim g}[w \leq x] + \mathbf{Pr}_{v \sim f}[v > y] \\
& = & \frac{r}{2} + \frac{r}{2} = r,
\end{eqnarray*}
which contradicts the definition $r = \mathbf{Pr}_{v \sim f, w \sim g}[w \leq v]$.
\qed
\end{proof}
With the above lemma in mind, we can prove our intended approximation factor for price $p^*$.

\begin{theorem}\label{thm:log-approx}
Let $(f,g)$ be any bilateral trade instance, and let $p^*$ be the price for $(f,g)$, as defined above. It holds that
\begin{equation*}
\text{OPT}_{f,g} \leq 4\log\left(\left\lceil\frac{2}{r}\right\rceil\right) \text{GFT}_{f,g}(p^*) .
%\in O\left(\log\left(\frac{1}{r}\right)\right) GFT_{f,g}(p^*).
\end{equation*}
\end{theorem}
\begin{proof}
We divide the proof into two cases, corresponding to the case distinction by which $p^*$ is defined. In the first case, it holds that $\mathbf{E}_{v \sim f, w \sim g}[(v-w)\mathbf{1}(w \leq v \wedge w > y)] \leq \text{OPT}_{f,g}/2$. This implies that $\mathbf{E}_{v \sim f, w \sim g}[(v-w)\mathbf{1}(w \leq v \wedge w \leq y)] \geq \frac{\text{OPT}_{f,g}}{2}$. For $i \in [\lceil\log(2/r)\rceil]$, let $E(i)$ denote the event
\begin{equation*}
\overline{F}^{-1}\left(\frac{1}{2^{i-1}}\right) \leq w \leq \overline{F}^{-1}\left(\frac{1}{2^i}\right) \wedge \overline{F}^{-1}\left(\frac{1}{2^{i-1}}\right) \leq v,
\end{equation*}
and let $f_{E(i)}$ and $g_{E(i)}$ be the probability distributions $f$ and $g$ conditioned on the event $E(i)$. Moreover, we define $\overline{F}_{E(i)}$ as the complementary cumulative distribution function of $f_{E(i)}$ and $G_{E(i)}$ as the cumulative distribution function of $g_{E(i)}$.

By definition, the price $p^*$ is the price among the prices $\{p_i : i \in [\lceil\log(2/r)\rceil]\}$ that achieves the highest gain from trade.
We proceed to show that we can write the expectation on the left hand side as a convex combination of $\lceil\log(2/r)\rceil$ values $T_1, \ldots, T_{\lceil\log(2/r)\rceil}$ (and $0$), such that $T_i \leq 2 \text{GFT}_{f_{E(i)},g_{E(i)}}(p_i)$ for all $i \in [\lceil\log(2/r)\rceil]$. Define 
\begin{equation*}
T_i = \mathbf{E}_{v \sim f, w \sim g}\left[(v-w)\mathbf{1}(w \leq v)\ \middle|\ E(i) \right] = \text{OPT}_{f_{E(i)},g_{E(i)}}.
\end{equation*}
We prove first that $T_i \leq 2 \text{GFT}_{f_{E(i)},g_{E(i)}}(p_i)$ for $i \in [\lceil\log(2/r)\rceil]$. 
We see that the function $\overline{F}_{E(i)}$ is at least $1/2$ on the interval $I_i = [\overline{G}^{-1}(1/2^{i-1}), \overline{G}^{-1}(1/2^i)]$, because $F$ is in the range $[1/2^{i-1}, 1/2^i]$ on interval $I_i$.
The function $G_{E(i)}$ crosses $\overline{F}_{E(i)}$ in $I_i$, and the price $p_i$ is, per definition, exactly the point in $I_i$ where the two functions are equal, i.e., the point $p_i \in I_i$ such that $\overline{F}_{E(i)}(p_i) = G_{E(i)}(p_i)$. Therefore, $1/2 \leq \overline{F}_{E(i)}(p_i) = G_{E(i)}(p_i)$.

Now we apply the first part of Theorem \ref{mainthm} to the instance $(f_{E(i)},g_{E(i)})$ with price $p_i$. Because of the above inequality, we know that the value of $q$ in Theorem \ref{mainthm} is at least $1/2$, so we obtain $\text{GFT}_{f_{E(i)},g_{E(i)}}(p_i) \geq (1/2)\text{OPT}_{f_{E(i)},g_{E(i)}} = (1/2)T_i$, as we wanted to show. We also derive that
\begin{align*}
& \mathbf{E}_{v \sim f, w \sim g}[(v-w)\mathbf{1}(w \leq v \wedge w \leq y)] \\
& \qquad \leq \sum_{i = 1}^{\lceil\log(2/r)\rceil} \mathbf{Pr}\left[\overline{F}^{-1}\left(\frac{1}{2^{i-1}}\right) \leq w \leq \overline{F}^{-1}\left(\frac{1}{2^i}\right) \wedge \overline{F}^{-1}\left(\frac{1}{2^{i-1}}\right) \leq v\right] T_i \\
& \qquad = \sum_{i = 1}^{\lceil\log(2/r)\rceil} \mathbf{Pr}_{v \sim f, w \sim g}[E(i)] T_i,
\end{align*}
is indeed a convex combination of $T_1, \ldots, T_{\lceil\log(2/r)\rceil}$ (and $0$). (Note that the inequality in the above derivation is an equality in case $\log(2/r)$ is an integer.)
It follows that 
\begin{align*}
& \text{OPT}_{f,g}/2 \leq 2 \sum_{i = 1}^{\lceil\log(2/r)\rceil} \mathbf{Pr}_{v \sim f, w \sim g}[E(i)] \text{GFT}_{f_{E(i)},g_{E(i)}}(p_i) \\
& \qquad = 2 \sum_{i = 1}^{\lceil\log(2/r)\rceil} \mathbf{Pr}_{v \sim f, w \sim g}[E(i)]\mathbf{E}_{v \sim f, w \sim g}[(v-w)\mathbf{1}(w \leq p_i \leq v)\ |\ E(i)] \\
& \qquad = 2 \sum_{i = 1}^{\lceil\log(2/r)\rceil} \mathbf{E}_{v \sim f, w \sim g}[(v-w)\mathbf{1}(w \leq p_i \leq v)\mathbf{1}(E(i))] \\
& \qquad \leq 2 \sum_{i = 1}^{\lceil\log(2/r)\rceil} \mathbf{E}_{v \sim f, w \sim g}[(v-w)\mathbf{1}(w \leq p_i \leq v)] \leq 2 \sum_{i = 1}^{\lceil\log(2/r)\rceil} \text{GFT}_{f,g}(p_i) \\
& \qquad \leq 2 \log\left(\left\lceil\frac{2}{r}\right\rceil\right) \max\left\{\text{GFT}_{f,g}(p_i) : i \in \left[\log\left(\left\lceil\frac{2}{r}\right\rceil\right)\right]\right\} = 2 \log\left(\left\lceil\frac{2}{r}\right\rceil\right) \text{GFT}_{f,g}(p^*),
\end{align*}
which proves the desired result for our first case.

For the second case, it holds that $\mathbf{E}_{v \sim f, w \sim g}[(v-w)\mathbf{1}(w \leq v \wedge w > y)] > \text{OPT}_{f,g}/2$. Then by Lemma \ref{caseslemma}, $\mathbf{E}_{v \sim f, w \sim g}[(v-w)\mathbf{1}(w \leq v \wedge v < x)] \leq (1/2) \text{OPT}_{f,g}$.
The analysis of this second case is from this point entirely analogous to that of the first case.
\qed
\end{proof}

Note that the approximation bound of $2/r$ that we established in the first section is better than the approximation bound of $4 \lceil \log(2/r) \rceil$ when $r$ is roughly greater than $0.05$. At $r = 0.05$, the approximation factor $4 \lceil \log(2/r)\rceil$ already takes a value around $20$. Hence, the result of this section is intended to provide theoretical insight into how the approximability of the gain from trade depends on $r$ asymptotically.

An (asymptotically) matching lower bound is given in the appendix of \cite{BD16}, which shows that $\Theta(\log(1/r))$ is asymptotically the best possible factor by which the optimal gain from trade that can always be approximated. We provide in Appendix \ref{apx:lowerbound} another independently discovered example of such a bilateral trade instance, in which the probability distributions of the buyer and the seller are discrete rather than continuous. Our example additionally shows that the approximation factor achievable by a fixed price mechanism can be no better than $\Omega(1/N)$, where $N$ is the number of points in the support of the distributions.

\section{A Fixed Price Double Auction}\label{sec:double-auctions}
We now turn to the double auction setting. Recall that in this setting there are $n \geq 1$ buyers and $m \geq 1$ sellers. The sellers each hold one item, and neither the buyers or the sellers are interested in holding more than one item. As before, we refer to $f$ for the probability distribution function from which the buyers' valuations are independently drawn, and to $g$ for the probability distribution from which the sellers' valuations are independently drawn. We denote the (random) valuation of buyer $i \in [n]$ by $v_i$ and the (random) valuation of seller $j \in [m]$ by $w_j$. See Section \ref{sec:prelims} for the definition. 

In order to present the definition of a \emph{fixed price mechanism} for the double auction setting, let us first introduce the concept of \emph{feasible pair}.

\begin{definition}
Let $(n,m,f,g)$ be an instance of a double auction setting, let $(v,w) \in \mathbb{R}^n\times \mathbb{R}^m$ be a valuation profile for the buyers and sellers, and let $p \in \mathbb{R}_{\geq 0}$. We call $(i,j) \in [n]\times[m]$ a \emph{feasible pair} with respect to profile $(v,w)$ and fixed price $p$ iff $v_i \geq p \geq w_j$. 
\end{definition}
\noindent Now, we can define a fixed price mechanism as follows.
\begin{definition}
We define a \emph{fixed price mechanism} $\mathbb{M}$ for a double auction setting $(n,m,f,g)$ as a direct revelation mechanism for which there is a price $p$ such that the mechanism selects a uniform random maximal subset of feasible pairs with respect to reported profile $(v,w)$ and $p$, and makes these pairs trade with each other. Moreover, for every selected trading pair $(i,j)$, the mechanism makes buyer $i$ pay an amount of $p$ to seller $j$. We refer to $p$ as the \emph{price} of $\mathbb{M}$.
\end{definition}
This is perhaps the most natural generalization of the notion of a fixed price mechanism that one may think of.
Please note that in a fixed price mechanism with price $p$, given a reported valuation profile $(v,w)$, the number of pairs that trade is always the minimum of $|\{v_i : v_i \geq p\}|$ and $|\{w_i : w_i \leq p\}|$. %This means that a fixed price mechanism always makes the maximum number of buyers and sellers trade with each other under the constraint that $p$ is in the buyers' and sellers' valuations.
 
It is easy to show that fixed price mechanisms clearly satisfy the three basic properties that we want:
\begin{theorem}\label{thm:properties}
For every double auction setting, every fixed price mechanism is ex-post IR, SBB, and DSIC.
\end{theorem}

Fixed price mechanisms have some additional advantanges. 
\begin{itemize}
 \item First, a fixed price mechanism is entirely symmetric: Each seller has the same expected utility, and each buyer has the same expected utility. The mechanism treats buyers with the same valuation entirely symetrically and does not break ties in favour of one over the other. This symmetricity is desirable from the point of view of fairness. 
 \item Secondly, the mechanism does not require the agents to fully reveal their entire valuation, since it can be implemented as a \emph{two-sided sequential posted price mechanism} \cite{doubleauctions}. Under such an implementation, the mechanisms goes over the buyers and sellers one by one. It proposes a take-it-or-leave-it price (equal to $p$, in this case) to each buyer and seller, which the buyers and sellers can choose to accept or reject. As soon as an accepting (buyer,seller)-pair is found, the mechanism lets this pair trade at price $p$. Taking a uniform random order of buyers and sellers will result in a random subset of feasible pairs who trade at price $p$, i.e., it will result in an implementation of the fixed price mechanism with price $p$. Under such an implementation, each buyer and seller has to reveal only one bit of information, which indicates whether her valuation is above or below $p$.
\end{itemize}

We aim to design a simple fixed price mechanism for which the gain from trade is a good approximation to the optimal gain from trade. 
The mechanism we use is as follows.
\begin{definition}
Given an instance $(n,m,f,g)$ of a double auction setting, let $\overline{p}$ be the price such that $n\overline{F}(\overline{p})= mG(\overline{p})$. We refer to the fixed price mechanism with price $\overline{p}$ as the \emph{balanced fixed price double auction}.
For ease of presentation, we refer to $\overline{F}(\overline{p})$ as $\overline{q}^B$ and we refer to $G(\overline{p})$ as $\overline{q}^S$. We denote by $\text{GFT}(\overline{p})$ the expected gain from trade achieved by the balanced fixed price double auction, and we denote by $\#T$ the expected number of trades that the balanced fixed price double auction generates.
\end{definition}
Observe that the balanced fixed price double auction is a generalization of the mechanism of Theorem \ref{mainthm} that achieves for the bilateral trade setting a $2/r$-approximation of the optimal gain from trade. We note that the value $\#T$ is by definition equal to $n\overline{q}^B = m \overline{q}^S$.

The main result we prove in this section is as follows.
\begin{theorem}\label{thm:doubleauctions}
For all $\epsilon \in [0,1]$, with probability at least $1 - 2/e^{\#T \epsilon^2 /2}$, the balanced fixed price double auction achieves a gain from trade that is at least $(1-\epsilon)$ times the expected optimal gain from trade.
\end{theorem}
Note that $\#T$, the expected number of trades of the balanced fixed price double auction, needs to exceed $2\ln(2)/\epsilon^2$ by any constant for the above theorem to yield a constant approximation guarantee. The value $\#T$ can be regarded as a property of the instance $(n,m,f,g)$ on which the mechanism is run, and is equal to the value where the functions $n\overline{F}$ and $mG$ cross each other. The requirement on $\#T$ is reasonably mild: For example, the above theorem says that when $\#T$ is at least $10$, the balanced fixed price double auction yields an expected gain from trade that is a $(< 4)$-approximation to the optimal gain from trade, by taking $\epsilon \approx 0.61$ (since $(1 - 2/e^{0.61^2 \cdot 5}) \cdot (1-0.61) > 0.25$). The theorem provides a constant approximation ratio for all instances where $\#T > 2\ln(2) \approx 1.38$, but grows unbounded as $\#T$ approaches $2\ln2$ from above.

There is an interesting interpretation of this theorem in terms of large markets: Observe that increasing the number of buyers or sellers in the market also increases $\#T$. In particular, by increasing both the number of buyers and the number of sellers simultaneously, $\#T$ grows unboundedly. From our theorem we may therefore infer that the balanced fixed price double auction approximation approximates the gain from trade by a factor that goes to $1$ as the market grows.

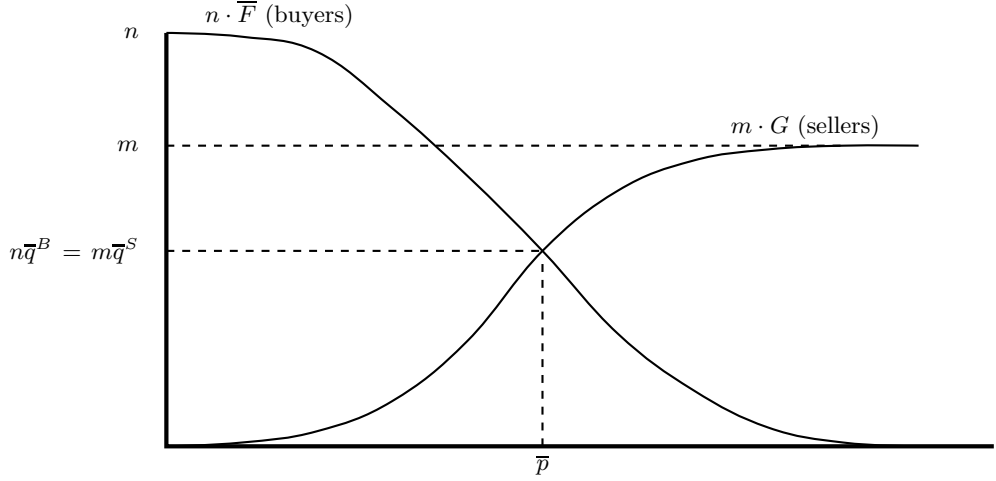
\begin{figure}
{\centering
\begin{tikzpicture}[scale=0.5]
\def\dotrad{3pt}
\tikzstyle{xxx}=[ultra thick,loosely dotted]
\tikzstyle{nome}=[rectangle,align=right,anchor=east,minimum width=1cm,text width=2cm]

\usetikzlibrary{patterns,decorations.pathreplacing}

\draw[ultra thick](0,11)--(0,0)--(22,0);
\draw[thick,dashed](0,8)--(18,8);

\draw[thick]plot[smooth,tension=0.7]coordinates
{(0,0)(2,0.1)(4,0.4)(6,1.2)(8,2.8)(10,5.2)(12,6.8)(14,7.6)(16,7.9)(18,8)(20,8)};

\draw[thick]plot[smooth,tension=0.7]coordinates
{(0,11)(2,10.9)(4,10.5)(6,9)(8,7.2)(10,5.2)(12,3)(14,1.5)(16,0.5)(18,0.1)(20,0)};

\draw[thick,dashed](10,0)--(10,5);\node at(10,-0.5){$\overline{p}$};
\draw[thick,dashed](0,5.2)--(10,5.2);

\node at(17,8.5){$m \cdot G$ (sellers)}; \node at(3,11.5){$n \cdot\overline{F}$ (buyers)};
\node[nome] at(-0.5,11){$n$};
\node[nome] at(-0.5,8){$m$};
\node[nome] at(-0.5,5.2){$n\overline{q}^B = m\overline{q}^S$};

\end{tikzpicture}
\par}
\caption{The graphs of the functions $n \cdot \overline{F}$ and $m \cdot G$ are shown.
They cross at price $\overline{p}$.}
\label{fig:gft}
\end{figure}

To prove the desired approximation property of the balanced fixed price double auction, we note that due to symmetry, we may assume that under the optimum allocation every buyer has the same a priori probability of trading with a seller, and every seller has the same a priori probability of trading with a buyer. This motivates the following definition.
\begin{definition}
For a double auction setting $(n,m,f,g)$ we define the values $q^S$ as the probability that any buyer receives an item under the optimum allocation, and we define $q^S$ as the probability that any seller loses her item under the optimum allocation. We define the prices $p^B$ and $p^S$ as the prices closest to $\overline{p}$ such that $\overline{F}(p^B) = q^B$ and $G(p^S) = q^S$. That is: $p^B$ is such that a buyers' probability of her valuation exceeding $p^B$ is equal to the probability of obtaining an item under the optimum allocation, and if there multiple such prices then $p^B$ is defined as the unique one closest to $\overline{p}$. Likewise, $p^S$ is such that a sellers' probability of her valuation being at most $p^S$ is equal to the probability of losing her item under the optimum allocation. Lastly, we let $\text{OPT}$ denote the expected gain from trade achieved by the optimum allocation. 
\end{definition}
The values $\text{GFT}(\overline{p})$, $\text{OPT}$, $\#T$, $\overline{p}$, $\overline{q}^B$, $\overline{q}^S$, $p$, $q^B$, and $q^S$ all depend (like $r$) on the instance $(n,m,f,g)$. We will leave this dependence implicit, as throughout this section there will be no ambiguity about the instance of the double auction setting that is being discussed.

%Further, let $p^B := G^{-1}(1 - q^B)$ and $p^S := F^{-1}(q^S)$, where $G^{-1}, F^{-1}$ are the inverse functions of $G, F$, e.g., for buyer $i$, $p^B$ is the price such that $\prob{v_i \geq p^B} = q^B$, and for seller $j$, $p^S$ is the price such that $\prob{w_j \leq p^S} = q^S$. 

%Now, we present some properties of the prices  $p^{B},  p^{S}, \overline{p}$.

\begin{lemma}\label{lem:optnumtrades}
For every instance $(n,m,f,g)$ of a double auction setting, the following property of the optimal allocation is satisfied.
\begin{equation*}
 nq^B = mq^S.
\end{equation*}
\end{lemma}

\noindent The following lemma states that our price $\overline{p}$ always lies in between $p^B$ and $p^S$.
\begin{lemma}\label{lem:ordering}
For every instance $(n,m,f,g)$ of a double auction setting, $p^{B} \geq \overline{p} \geq p^{S}$ or $p^{S} \geq \overline{p} \geq p^{B}$.
\end{lemma}

\noindent The following lemma provides a useful bound on $\text{OPT}$
\begin{lemma}\label{prob:gftopt}
For every instance $(n,m,f,g)$ of a double auction setting, it holds that 
\begin{equation*}
\text{OPT} \leq n q^B \mathbf{E}[v_1\ |\ v_1 \geq p^B] - m q^S \mathbf{E}[w_1\ |\ w_1 \leq p^S].
\end{equation*}
\end{lemma}

We will use the following technical Lemma to bound $\text{OPT}$ further.
\begin{lemma}\label{lem:bound1}
For every instance $(n,m,f,g)$ of a double auction setting, it holds that
\begin{equation*}
n q^B \mathbf{E}[v_1\ |\ v_1 \geq p^B] - m q^S \mathbf{E}[w_1\ |\ w_1 \leq p^S] \leq n \overline{q}^B \mathbf{E}[v_1\ |\ v_1 \geq {\overline p}] - m \overline{q}^S \mathbf{E}[w_1\ |\ w_1 \leq {\overline p}]
\end{equation*}
\end{lemma}

%The upper bound follows: 

%\begin{corollary}\label{coro:bound}
%\begin{eqnarray*}
%OPT &\leq& n q^B \expected{v_1\ |\ v_1 \geq p^B} - m q^S \expected{w_1\ |\ w_1 \leq p^S} \\
% &\leq& n {\overline q}^B \expected{v_1\ |\ v_1 \geq {\overline p}} - m \overline{q}^S \expected{w_1\ |\ w_1 \leq {\overline p}}
%\end{eqnarray*}
%\end{corollary}

We can now prove our main approximation result for the balanced fixed price double auction.
\begin{proof}[of Theorem \ref{thm:doubleauctions}]
Let us define the following random variables: $B_{i}$ is a random variable that is equal to $1$ if $v_{i} \geq \overline{p}$ and $0$ otherwise, similarly, $S_{j}$ is a random variable that is equal to $1$ if $w_{j} \leq \overline{p}$ and $0$ otherwise. Now, let $B = \sum_{i} B_{i}$ and $S= \sum_{j} S_{j}$. Moreover, let $B_i^*$ denote the event that $i$ trades under the balanced fixed price double auction, and similarly let $B_j^*$ be the event that $j$ trades under the balanced fixed price double auction. Denote by $E$ the event that the balanced fixed price double auction lets at least $(1-\epsilon)n\overline{q}^B = (1-\epsilon)m\overline{q}^S$ pairs of buyers and sellers trade. That is $E$ denotes the event that $B \geq (1-\epsilon)n \overline{q}^{B} \wedge S \geq (1-\epsilon) m \overline{q}^{S}$. First, we derive the following bound on the gain from trade of the balanced fixed price double auction conditioned on $E$. 
%Let $i \tradesp j$ denote the event that the balanced fixed price double auction lets buyer $i$ trade with seller $j$.
\begin{eqnarray*}
& & \mathbf{E}\left[\sum_{i = 1}^n (v_i - \overline{p})\mathbf{1}(B_i^*) + \sum_{j=1}^m(\overline{p} - w_j)\mathbf{1}(S_j^*)\ \middle|\ E \right] \\
& = & \sum_{i = 1}^n \mathbf{E}\left[v_i - \overline{p}\ \middle|\ E \wedge B_i^*\right]\mathbf{Pr}\left[B_i^*\ \middle|\ E \right] + \sum_{j=1}^m \mathbf{E}\left[\overline{p} - w_j\ \middle|\ E \wedge S_j^* \right]\mathbf{Pr}\left[S_j^*\ \middle|\ E \right] \\
& = & \sum_{i = 1}^n \mathbf{E}\left[v_i - \overline{p}\ \middle|\ v_i \geq \overline{p} \right]\mathbf{Pr}\left[B_i^*\ \middle|\ E \right] + \sum_{j=1}^m \mathbf{E}\left[\overline{p} - w_j\ \middle|\ w_j \leq \overline{p} \right]\mathbf{Pr}\left[S_j^*\ \middle|\ E \right] \\
& \geq & \sum_{i = 1}^n \mathbf{E}\left[v_i - \overline{p}\ \middle|\ v_i \geq \overline{p} \right]\frac{n\overline{q}^B(1-\epsilon)}{n} + \sum_{j=1}^m \mathbf{E}\left[\overline{p} - w_j\ \middle|\ w_j \leq \overline{p} \right]\frac{m\overline{q}^S(1-\epsilon)}{m} \\
& = & n \mathbf{E}\left[v_1 - \overline{p}\ \middle|\ v_1 \geq \overline{p} \right]\overline{q}^B(1-\epsilon) + m \mathbf{E}\left[\overline{p} - w_1\ \middle|\ w_1 \leq \overline{p} \right]\overline{q}^S(1-\epsilon). \\
& = & n \mathbf{E}\left[v_1\ \middle|\ v_1 \geq \overline{p} \right]\overline{q}^B(1-\epsilon) - m \mathbf{E}\left[w_1\ \middle|\ w_1 \leq \overline{p} \right]\overline{q}^S(1-\epsilon). \\
\end{eqnarray*}

We bound as follows the probability of the event $E$ happening.
\begin{eqnarray*} 
\mathbf{Pr}[E] & = & 1 - \mathbf{Pr}[B < (1-\epsilon)n \overline{q}^{B} \vee S < (1-\epsilon) m \overline{q}^{S}] \\ 
& \geq & 1 - \mathbf{Pr}[B < (1-\epsilon)n \overline{q}^{B}] - \mathbf{Pr}[S < (1-\epsilon) m \overline{q}^{S}]\\
& \geq & 1 - \frac{2}{e^{n\overline{q}^B \epsilon^2 / 2}} \\
& \geq & 1 - \frac{2}{e^{\#T \epsilon^2 / 2}} \\
\end{eqnarray*}
The second inequality follows from applying a standard Chernoff bound, which can be done since $B$ and $S$ are sums of $\{0,1\}$ independent random variables.
%Notice that in order to have a meaningful probability the following inequalities must hold: $n\overline{q}^{B} \geq \log(2)\frac{2}{\epsilon^{2}}$ and $m\overline{q}^{S} \geq \log(2)\frac{2}{\epsilon^{2}}$.

We conclude that with probability at least $1 - 2/e^{\#T \epsilon^2 / 2}$ the mechanism obtains at least $(1-\epsilon) \#T $ trades at price $\overline{p}$. Thus, using Lemmas \ref{prob:gftopt} and \ref{lem:bound1} we conclude that with probability at least $1 - 2/e^{\#T \epsilon^2 / 2}$ the gain from trade of the balanced fixed price double auction is at least
\begin{equation*} 
(1-\epsilon) (n \overline{q}^B \mathbf{E}[v_1\ |\ v_1 \geq \overline{p}] - m \overline{q}^S \mathbf{E}[w_1\ |\ w_1 \leq \overline{p}]) \geq (1-\epsilon) \text{OPT},
\end{equation*}
which completes the proof.
\qed
\end{proof}

\begin{paragraph}{\bf{Acknowledgements.}}
We thank Tim Roughgarden for helpful discussions at the early stages of this work.
\end{paragraph}

\newpage

\bibliographystyle{plain}
\bibliography{gainfromtrade}

\appendix

\section{Asymptotically Tight Lower Bound on the Gain From Trade Achievable by Fixed Brice Mechanisms}
\label{apx:lowerbound}
The following family of examples show how the ratio between the gain from trade of the optimal gain from trade and the best gain from trade achievable by a fixed price mechanism can be as high as $\Omega(\log(1/r))$, and moreover can grow linearly with the support size of the buyer's and seller's distributions.

%We start from the fact that if a mechanism has to satisfy IR, DSIC, and SBB, then it must be a fixed price mechanism, i.e., the mechanism fixes a price $p$ and posts it to the buyer and the seller. For the remainder of the section, we use $v, w$ respectively to denote the buyer's and seller's valuation. We assume these being drawn from distributions $g, f$, where $g(x) = \prob{v = x}$ and $f(x) = \prob{w = x}$. The expected \emph{gain from trade} (\GFT) of a fixed price $p$ mechanism $\mathbb{M}_p$ is:
%\[
%\GFT(\mathbb{M}_p) := \sum_{x \in \supp{f} : x \leq p} \sum_{y \in \supp{g} : y \geq p} (y - x) f(x) g(y).
%\]
%The fixed price mechanism $\mathbb{M}^*_p$ that maximizes the \GFT \ can be found by solving the following expression:
%\[
%\GFT(\mathbb{M}^*_p) := \max_{p \in \supp{f}} \left\{\sum_{x \in \supp{f} : x \leq p} \sum_{y \in \supp{g} : y \geq p} (y - x) f(x) g(y) \right\}.
%\]
%The expected optimal \GFT \ lets the buyer and seller trade whenever $w \leq v$:
%\[
%OPT := \sum_{x \in \supp{f}} \sum_{y \in \supp{g} : y \geq x} (y - x) f(x) g(y).
%\]
%With the following example, we desire to show that the ratio between $\GFT(\mathbb{M}^*_p)$ and
%$OPT$ can go as rapidly to zero as $|\supp{f}| = |\supp{g}|$ grows.

\begin{example}
\label{example:lower_bound_gft}
Let $N \in \mathbb{N}$ and let $\epsilon$ be any number in $[5/36, 1)$. Denote by $f_N$ and $g_N$ the buyer's and seller's valuation distribution, and let the supports of these distributions be respectively $\supp{f_N} = \{1+\epsilon, 2+\epsilon, \ldots, N+\epsilon\}$ and $\supp{g_N} = [N]$. Let $\alpha = \sum_{x = 0}^{N-1}10^{-x} = \frac{10}{9} (1 - 10^{-N})$. Then, for all $w \in \supp{g_N}, g_N(w) := 10^{w-N} / \alpha$ and for all $v \in \supp{f_N}, f(v) := 10^{-(v-1-\epsilon)} / \alpha$. Figure \ref{fig:example_gft} provides an illustration of $(f_5,g_5)$.
\end{example}

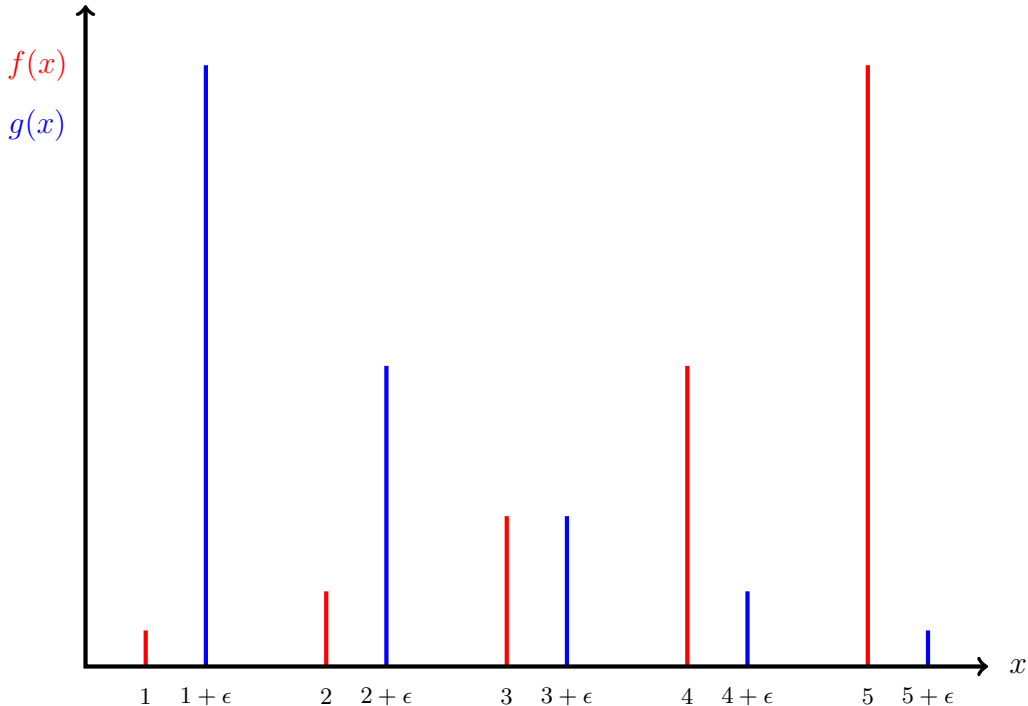
\begin{figure}
{\centering
\begin{tikzpicture}[scale=0.8]
\def\dotrad{3pt}
\tikzstyle{xxx}=[ultra thick,loosely dotted]
\tikzstyle{nome}=[rectangle,align=right,anchor=east,minimum width=1cm,text width=1.2cm]

\draw[color=red,ultra thick](1,0)--(1,0.6);
\draw[color=red,ultra thick](4,0)--(4,1.25);
\draw[color=red,ultra thick](7,0)--(7,2.5);
\draw[color=red,ultra thick](10,0)--(10,5);
\draw[color=red,ultra thick](13,0)--(13,10);

\draw[color=blue,ultra thick](2,0)--(2,10);
\draw[color=blue,ultra thick](5,0)--(5,5);
\draw[color=blue,ultra thick](8,0)--(8,2.5);
\draw[color=blue,ultra thick](11,0)--(11,1.25);
\draw[color=blue,ultra thick](14,0)--(14,0.6);

\draw[ultra thick,<->](0,11)--(0,0)--(15,0);
\node at(1,-0.5){$1$};\node at(2,-0.5){$1+\epsilon$};
\node at(4,-0.5){$2$};\node at(5,-0.5){$2+\epsilon$};
\node at(7,-0.5){$3$};\node at(8,-0.5){$3+\epsilon$};
\node at(10,-0.5){$4$};\node at(11,-0.5){$4+\epsilon$};
\node at(13,-0.5){$5$};\node at(14,-0.5){$5+\epsilon$};

\node[color=red]at(-0.8,10){\large $f(x)$};
\node[color=blue]at(-0.8,9){\large $g(x)$};
\node at(15.5,0){\large $x$};

\end{tikzpicture}
\par}
\caption{A representation of the above example for $N=5$. The probability masses of $f_5$ are displayed in blue while those of $f_5$ are displayed in red.}\label{fig:example_gft}
\end{figure}

The remainder of this section comprises a proof our claimed lower bounds on the achievable approximation factor, for the above family of examples. Consider $N$ to as a fixed natural number for the remaining part of this section.
The distributions of $f_N$ and $g_N$ are discretised versions of exponential distributions. The role of the scaling parameter $\alpha$ is solely to make sure that they are feasible distributions, i.e., that the probability masses sum up to $1$. Observe that for the buyer (and symmetrically, for the seller) it holds that for any $v \in \supp{f_N}$, $f_N(v) > \sum_{v' \in \supp{f_N} : v' > v} f_N(v')$, i.e., the probability mass at any point $v$ in the support of $f_N$ is larger than the total probability mass of all points $v' > v$.

Observe that for the above example we have that the gain from trade achieved by a fixed price mechanism with fixed price $p$ satisfies
\begin{equation*}
\text{GFT}_{f_N,g_N}(p) \geq \frac{10^{1-N}}{\alpha^2} \sum_{w=1}^p \sum_{v=p}^N (v+\epsilon-w) 10^{w-v} 
\end{equation*}
and the optimum gain from trade equals
\begin{equation*}
\text{OPT}_{f_N,g_N} = \frac{10^{1-N}}{\alpha^2} \sum_{w=1}^N \sum_{v=w}^N (v+\epsilon-w) 10^{w-v}.
\end{equation*}
We will first show that the ratio $\text{GFT}_{f_N,g_N}(p) / \text{OPT}_{f_N,g_N}$ can be as bad as $4/N$, where $N$ is the support size of $f_N$ and $g_N$.

\begin{lemma}
\label{lemma:bound_on_buyer}
For any $w \in \{1, \ldots, p\}, p \in [N],$ and $\epsilon \geq 5/36$, it holds that
\begin{equation*}
\sum_{v=p}^N (v+\epsilon-w) 10^{w-v} \leq 2(p+\epsilon-w) 10^{v-p}.
\end{equation*}
\end{lemma}

\begin{proof}
Let $z = p - x$. Then we derive:
\begin{eqnarray*}
\sum_{v=p}^N (v+\epsilon-w) 10^{w-v} & = & (p+\epsilon-w) 10^{w-p} + \sum_{\ell=1}^{N-p} (p+\epsilon-w+\ell) 10^{w-p-\ell}\\
& = & (z+\epsilon) 10^{-z} + \sum_{\ell=1}^{N-p} (z+\epsilon+\ell) 10^{-z-\ell}.
\end{eqnarray*}
In the first equality we took $y=p$ out of the sum and expressed the remaining terms of the sum as $p+\ell$. To prove the statement it is enough to show that the last summation of the last expression is at most $(z+\epsilon) 10^{-z}$.
\begin{eqnarray*}
\sum_{\ell=1}^{N-p} (z+\epsilon+\ell) 10^{-z-\ell} & \leq & 10^{-z} \left(\sum_{\ell=1}^{\infty}(z+\epsilon + \ell) 10^{-\ell}\right)\\
& = & 10^{-z} \left(\frac{1}{9^2} (9(z+\epsilon) +10)\right).
\end{eqnarray*}
Now observe that the following equivalences hold.
\begin{equation*}
10^{-z} \left(\frac{z+\epsilon+10/9}{9}\right) \leq (z+\epsilon) \Longleftrightarrow z+\epsilon \geq \frac{5}{36} \Longleftrightarrow \epsilon \geq \frac{5}{36} - z.
\end{equation*}
Since $z = p-w$ and $w \leq p$, it holds that $z \geq 0$. Hence, any choice of $\epsilon \geq 5/36$ satisfies the inequality.
\qed
\end{proof}

\begin{lemma}
\label{lemma:bound_on_alg}
For any $p \in [N]$ and $\epsilon \geq 5/36$,
\begin{equation*}
\sum_{w=1}^p \sum_{v=p}^N (v+\epsilon-w) 10^{w-v} \leq 4\epsilon.
\end{equation*}
\end{lemma}

\begin{proof}
\begin{eqnarray*}
\sum_{w=1}^p \sum_{v=p}^N (v+\epsilon-w) 10^{w-v} & \leq & \sum_{w=1}^p 2 (p+\epsilon-w) 10^{w-p}\\
& = & 2\epsilon + 2 \sum_{w=1}^{p-1} 2 (p+\epsilon-w) 10^{w-p}\\
& \leq & 2\epsilon + 2 \sum_{w=1}^{\infty} 2 (p+\epsilon-w) 10^{w-p}\\
& = & 2\epsilon + 2 \left(\frac{9\epsilon + 10}{9^2}\right).
\end{eqnarray*}
As showed in \refLemma{lemma:bound_on_buyer},
\begin{equation*}
\frac{9\epsilon + 10}{9^2} \leq \epsilon \Longleftrightarrow \epsilon \geq \frac{5}{36}.
\end{equation*}
\qed
\end{proof}

The following lemma gives an upper bound on the expected optimum gain from trade.
\begin{lemma}
\label{remark:bound_on_opt}
For all $N \in \mathbb{N}$,
\begin{equation*}
\sum_{w=1}^N \sum_{v=w}^N (v+\epsilon-w) 10^{w-x} \geq N \epsilon.
\end{equation*}
\end{lemma}

\begin{proof}
If we only sum on $v$'s that are equal to $w$, we certainly lower-bound the expression and get that
\begin{equation*}
\sum_{w=1}^N \sum_{v=w}^N (v+\epsilon-w) 10^{v-w} \geq \sum_{w=1}^N \epsilon = N\epsilon.
\end{equation*}
\qed
\end{proof}

\noindent The next proposition follows from Lemmas \ref{lemma:bound_on_alg} and \ref{remark:bound_on_opt}.
\begin{proposition}
Let $(f_N,g_N)$ be the defined as in Example \ref{example:lower_bound_gft}. For every fixed price $p$, the ratio of  $\text{OPT}_{f_N,g_N}$ and $\text{GFT}_{f_N,g_N}(p)$ is at least $N/4$, where $N$ is the support size of $f$ and $g$.
\end{proposition}

\noindent It remains to show that $N \in \Omega(\log(1/r))$. This follows from the following bound on $r$.

\begin{lemma}
\begin{equation*}
 N \geq \log\left(\frac{1}{r}\right).
\end{equation*}
\end{lemma}
\begin{proof}
The value $r$ is defined as the probability of the buyer's valuation exceeding that of the seller. Therefore, we can write $r$ as follows.
\begin{eqnarray*}
r & = & \sum_{v \in \supp{f_N}} \sum_{w \in \supp{g_N} : w < v} f(v)g(w) \\
& = &\frac{1}{(100/81)(1 - 10^{-N})^2} 10^{1-N+\epsilon} \sum_{(v,w) \in [N]^2 : w \leq v} 10^{w-v} \\
& \geq & \frac{81}{100} 10^{1-N+\epsilon} \sum_{i=0}^{N-1} (N-i)10^{-i} \\
& = & \frac{81}{100} 10^{1-N+\epsilon} \frac{10}{81}(9N + 10^-N + 1) \\
& \geq & \frac{1}{10} 10^{1-N+\epsilon} \\
& = & 10^{-N+\epsilon}.
\end{eqnarray*}
Therefore, $1/r \leq 10^{N - \epsilon}$ and $\log_{10}(1/r) \leq N - \epsilon \leq N$.
\qed
\end{proof}
Combining the above proposition and lemma, we conclude that for every fixed price $p$, the ratio $\text{OPT}_{f_N,g_N}/\text{GFT}_{f_N,g_N}(p)$ is at least $\Omega(\log(1/r))$.

\section{Proofs}

\subsection{Proof of Theorem \ref{mainthm}}
%\begin{proof}[of Theorem \ref{mainthm}]
Without loss of generality, we assume that the distributions $f$ and $g$ are discrete, and we write $\text{supp}(f)$ and $\text{supp}(g)$ to denote the supports of $f$ and $g$ respectively.

The optimal gain from trade can be written as the following summation.
\begin{equation*}
\text{OPT}_{f,g} = \sum_{(v,w) \in \text{supp}(f) \times \text{supp}(g) : v \geq w} f(v)g(w)(v - w)
\end{equation*}
We split this summation up into three parts: We define $\text{MGFTL}_{f,g}(p)$ as
\begin{equation*}
\sum_{(v,w) \in \text{supp}(f) \times \text{supp}(g) :  w \leq v < p} f(v)g(w)(v-w),
\end{equation*}
and we define $\text{MGFTR}_{f,g}(p)$ as
\begin{equation*}
 \sum_{(v,w) \in \text{supp}(f) \times \text{supp}(g) : p < w \leq v} f(v)g(w)(v-w),
\end{equation*}
so that $\text{OPT}_{f,g} = \text{MGFTL}_{f,g}(p) + \text{GFT}_{f,g}(p) + \text{MGFTR}_{f,g}(p)$.\footnote{We note that ``MGFTL and MGFTR'' are intended to stand for ``Missed Gain From Trade on the Left'' and ``Missed Gain From Trade on the Right'', respectively.}
Moreover, we split $\text{GFT}_{f,g}(p)$ up into two parts: We define $\text{GFTL}_{f,g}(p)$ as 
\begin{equation*}
\sum_{(v,w) \in \text{supp}(f) \times \text{supp}(g) : w \leq p \leq v} f(v)g(w)(p-w) \leq q \sum_{w \in \text{supp}(g) : w \leq p} g(w)(p-w)
\end{equation*}
and we define $\text{GFTR}_{f,g}(p)$ as 
\begin{equation*}
\sum_{(v,w) \in \text{supp}(f) \times \text{supp}(g) : v \leq p \leq w} f(v)g(w)(v-p) \leq q \sum_{v \in \text{supp}(f) : p \leq v} f(v)(v-p),
\end{equation*}
so that $\text{GFT}_{f,g}(p) = \text{GFTL}_{f,g}(p) + \text{GFTR}_{f,g}(p)$.
 
Using this fact, we derive the following.
\begin{eqnarray*}
\text{MGFTL}_{f,g}(p) & = & \sum_{(v,w) \in \text{supp}(f) \times \text{supp}(g) : w \leq v < p} f(v)g(w)(v-w) \\
& \leq & \sum_{(v,w) \in \text{supp}(f) \times \text{supp}(g) : w \leq v < p} f(v)g(w)(p-w) \\
& = & \sum_{w \in \text{supp}(g) : w \leq p} g(w)(p-w) \sum_{v \in \text{supp}(f) : w \leq v < p} f(v) \\
& \leq & \sum_{w \in \text{supp}(g) : w < p} g(w)(p-w) (1-q) \\
& = & \frac{1-q}{q}\text{GFTL}_{f,g}(p),
\end{eqnarray*}
and analogously, we obtain
\begin{eqnarray*}
\text{MGFTR}_{f,g}(p) & = & \sum_{(v,w) \in \text{supp}(f) \times \text{supp}(g) : p < w \leq v} f(v)g(w)(v-w) \\
& \leq & \sum_{(v,w) \in \text{supp}(f) \times \text{supp}(g) : p < w \leq v} f(v)g(w)(w-p) \\
& = & \sum_{v \in \text{supp}(f) : v \geq p} f(v)(v-p) \sum_{w \in \text{supp}(g) : p < w \leq v} g(w) \\
& \leq & \sum_{v \in \text{supp}(f) : v \geq p} f(v)(v-p) (1-q) \\
& = & \frac{1-q}{q} \text{GFTR}_{f,g}(p).
\end{eqnarray*}
Therefore, 
\begin{eqnarray*} 
\text{OPT}_{f,g} & = & \text{MGFTL}_{f,g}(p) + \text{GFT}_{f,g}(p) + \text{MGFTR}_{f,g}(p) \\ 
& \leq & \frac{1-q}{q} \text{GFTL}_{f,g}(p) + \text{GFT}_{f,g}(p) + \frac{1-q}{q} \text{GFTR}_{f,g}(p) \\
& = & \left(\frac{1-q}{q} + 1\right) \text{GFT}_{f,g}(p) \\
& = & \frac{1}{q} \text{GFT}_{f,g}(p),
\end{eqnarray*}
which establishes (\ref{eq:1}).
To obtain (\ref{eq:2}), set $p$ such that $q$ is maximized, and note that 
\begin{eqnarray*}
r & = & \mathbf{Pr}_{v \sim f, w \sim g}[v \geq w] \\
& = & \mathbf{Pr}_{v \sim f, w \sim g}[ w \leq v \leq p ] + \mathbf{Pr}_{v \sim f, w \sim g}[ w \leq p \leq v] + \mathbf{Pr}_{v \sim f, w \sim g}[ p \leq w \leq v] \\
& \leq & q(1-q) + q^2 + q(1-q) \\
& \leq & 2q.
\end{eqnarray*}
Therefore,
\begin{equation*}
\text{OPT}_{f,g} \leq \frac{1}{q} \text{GFT}_{f,g}(p) \leq \frac{2}{r} \text{GFT}_{f,g}(p),
\end{equation*}
which establishes (\ref{eq:2}).
%\qed \end{proof}

\subsection{Proof of Theorem \ref{thm:log-approx}}

%\begin{proof}
We divide the proof into two cases, corresponding to the case distinction by which $p^*$ is defined. In the first case, it holds that $\mathbf{E}_{v \sim f, w \sim g}[(v-w)\mathbf{1}(w \leq v \wedge w > y)] \leq \text{OPT}_{f,g}/2$. This implies that: 
\begin{equation*}
\mathbf{E}_{v \sim f, w \sim g}[(v-w)\mathbf{1}(w \leq v \wedge w \leq y)] \geq \frac{\text{OPT}_{f,g}}{2}.
\end{equation*}

For $i \in [\lceil\log(2/r)\rceil]$, let $E(i)$ denote the event
\begin{equation*}
\overline{G}^{-1}\left(\frac{1}{2^{i-1}}\right) \leq w \leq \overline{G}^{-1}\left(\frac{1}{2^i}\right) \wedge \overline{G}^{-1}\left(\frac{1}{2^{i-1}}\right) \leq v,
\end{equation*}
and let $f_{E(i)}$ and $g_{E(i)}$ be the probability distributions $f$ and $g$ conditioned on the event $E(i)$. Moreover, we define $\overline{F}_{E(i)}$ as the complementary cumulative distribution function of $f_{E(i)}$ and $G_{E(i)}$ as the complementary cumulative distribution function of $g_{E(i)}$.

By definition, the price $p^*$ is the price among the prices $\{p_i : i \in [\lceil\log(2/r)\rceil]\}$ that achieves the highest gain from trade.
We proceed to show that we can write the expectation on the left hand side as a convex combination of $\lceil\log(2/r)\rceil$ values $T_1, \ldots T_{\lceil\log(2/r)\rceil}$ (and $0$), such that $T_i \leq 2 \text{GFT}_{f_{E(i)},g_{E(i)}}(p_i)$ for all $i \in [\lceil\log(2/r)\rceil]$. Define 
\begin{eqnarray*}
T_i & = & \mathbf{E}_{v \sim f, w \sim g}\left[(v-w)\mathbf{1}(w \leq v)\ \middle|\ \overline{F}^{-1}\left(\frac{1}{2^{i-1}}\right) \leq w \leq \overline{F}^{-1}\left(\frac{1}{2^i}\right) \wedge \overline{F}^{-1}\left(\frac{1}{2^{i-1}}\right) \leq v\right] \\
& = & \text{OPT}_{f_{E(i)},g_{E(i)}}.
\end{eqnarray*}
We prove first that $T_i \leq 2 \text{GFT}_{f_{E(i)},g_{E(i)}}(p_i)$ for $i \in [\lceil\log(2/r)\rceil]$. 
We see that the function $\overline{F}_{E(i)}$ is at least $1/2$ on the interval $I_i = [\overline{G}^{-1}(1/2^{i-1}), \overline{G}^{-1}(1/2^i)]$, because $F$ is in the range $[1/2^{i-1}, 1/2^i]$ on interval $I_i$.
The function $G_{E(i)}$ crosses $\overline{F}_{E(i)}$ in $I_i$, and the price $p_i$ is, per definition, exactly the point in $I_i$ where the two functions are equal, i.e., the point $p_i \in I_i$ such that $\overline{F}_{E(i)}(p_i) = G_{E(i)}(p_i)$. Therefore, 
\begin{equation*}
\frac{1}{2} \leq \overline{F}_{E(i)}(p_i) = G_{E(i)}(p_i)
\end{equation*}
Now we apply the first part of Theorem \ref{mainthm} to the instance $(f_{E(i)},g_{E(i)})$ with price $p_i$. Because of the above inequality, we know that the value of $q$ in Theorem \ref{mainthm} is at least $1/2$, so we obtain:
\begin{equation*}
\text{GFT}_{f_{E(i)},g_{E(i)}}(p_i) \geq \frac{1}{2} \text{OPT}_{f_{E(i)},g_{E(i)}} = \frac{1}{2} T_i,
\end{equation*}
as we wanted to show. We also derive that
\begin{align*}
& \mathbf{E}_{v \sim f, w \sim g}[(v-w)\mathbf{1}(w \leq v \wedge w \leq y)] \\
& \qquad \leq \sum_{i = 1}^{\lceil\log(2/r)\rceil} \mathbf{Pr}\left[\overline{F}^{-1}\left(\frac{1}{2^{i-1}}\right) \leq w \leq \overline{F}^{-1}\left(\frac{1}{2^i}\right) \wedge \overline{F}^{-1}\left(\frac{1}{2^{i-1}}\right) \leq v\right] T_i \\
& \qquad = \sum_{i = 1}^{\lceil\log(2/r)\rceil} \mathbf{Pr}_{v \sim f, w \sim g}[E(i)] T_i,
\end{align*}
is indeed a convex combination of $T_1, \ldots, T_{\lceil\log(2/r)\rceil}$ (and $0$). (Note that the inequality in the above derivation is an equality in case $\log(2/r)$ is an integer.)
It follows that 
\begin{eqnarray*}
\text{OPT}_{f,g}/2 & \leq & 2 \sum_{i = 1}^{\lceil\log(2/r)\rceil} \mathbf{Pr}_{v \sim f, w \sim g}[E(i)] \text{GFT}_{f_{E(i)},g_{E(i)}}(p_i) \\
& = & 2 \sum_{i = 1}^{\lceil\log(2/r)\rceil} \mathbf{Pr}_{v \sim f, w \sim g}[E(i)]\mathbf{E}_{v \sim f, w \sim g}[(v-w)\mathbf{1}(w \leq p_i \leq v)\ |\ E(i)] \\
& = & 2 \sum_{i = 1}^{\lceil\log(2/r)\rceil} \mathbf{E}_{v \sim f, w \sim g}[(v-w)\mathbf{1}(w \leq p_i \leq v)\mathbf{1}(E(i))] \\
& \leq & 2 \sum_{i = 1}^{\lceil\log(2/r)\rceil} \mathbf{E}_{v \sim f, w \sim g}[(v-w)\mathbf{1}(w \leq p_i \leq v)] \\
& \leq & 2 \sum_{i = 1}^{\lceil\log(2/r)\rceil} \text{GFT}_{f,g}(p_i) \\
& \leq & 2 \log\left(\left\lceil\frac{2}{r}\right\rceil\right) \max\left\{\text{GFT}_{f,g}(p_i) : i \in \left[\log\left(\left\lceil\frac{2}{r}\right\rceil\right)\right]\right\} \\
& = & 2 \log\left(\left\lceil\frac{2}{r}\right\rceil\right) \text{GFT}_{f,g}(p^*),
\end{eqnarray*}
which proves the desired result for our first case.

For the second case, it holds that $\mathbf{E}_{v \sim f, w \sim g}[(v-w)\mathbf{1}(w \leq v \wedge w > y)] > \text{OPT}_{f,g}/2$. Then by Lemma \ref{caseslemma}, 
\begin{equation*}
\mathbf{E}_{v \sim f, w \sim g}[(v-w)\mathbf{1}(w \leq v \wedge v < x)] \leq \frac{\text{OPT}_{f,g}}{2}
\end{equation*}
The analysis of this second case is from this point entirely analogous to that of the first case.
%\qed\end{proof}

\subsection{Proof of Theorem \ref{thm:properties}}
Let $p$ be the fixed price at which the mechanism trades.
If the valuation of a buyer is less than $p$, the buyer has a negative expected utility if she reports more than $p$, and otherwise she has a utility of $0$.
If the valuation of a buyer is at least $p$ then reporting anything above $p$ will yield her the same expected non-negative utility, and bidding below $p$ will yield her a utility of $0$.
A symmetric argument holds for the sellers. Thus, for all agents there is never incentive to misreport her valuation, which proves the DSIC property. Moreover, an agent's expected valuation is always non-negative when she bids truthfully, which proves the ex-post IR property. Lastly, since for every trading pair, the seller receives exactly the amount paid by the buyer, so the SBB property is trivially satisfied.
%\qed \end{proof}

\subsection{Proof of Lemma \ref{lem:optnumtrades}}
Note that $nq^B$ expresses the expected number of buyers that trade under the optimum allocation, and $mq^S$ expresses the expected number of sellers that trade under the optimum allocation.
Since the number of buyers that trade is equal to the number of sellers that trade for every valuation profile, these expected values must be equal.
%\qed \end{proof}

\subsection{Proof of Lemma \ref{lem:ordering}}
This proof is by contradiction. Let us separately examine the two cases that we want to rule out.

First, assume by contradiction that $p^{B} > \overline{p}$ and $p^{S} > \overline{p}$. Then, it holds that $\overline{q}^B = \overline{F}(\overline{p}) \geq \overline{F}(p^B) = q^B$ and $q^S = G(p^S) \geq G(\overline{p}) = \overline{q}^b$. Thus, $n \overline{q}^{B} \geq n q^{B} = m q^{S} \geq m \overline{q}^{S} = n \overline{q}^S$, where the first equality holds by Lemma \ref{lem:optnumtrades}. Hence, $q^{S} = \overline{q}^S$, but by definition of $p^S$ that means that $p^S = \overline{p}$ which is a contradiction.
	
Similarly, assume by contradiction that $p^{B} < \overline{p}$ and $p^{S} < \overline{p}$. Then, it holds that $\overline{q}^B = \overline{F}(\overline{p}) \leq \overline{F}(p^B) = q^B$ and $q^S = G(p^S) \leq G(\overline{p}) = \overline{q}^b$. Thus, $n \overline{q}^{B} \leq n q^{B} = m q^{S} \leq m \overline{q}^{S} = n \overline{q}^B$, where the first equality holds by Lemma \ref{lem:optnumtrades}. Hence, $q^{S} = \overline{q}^S$, but by definition of $p^S$ that means that $p^S = \overline{p}$ which is a contradiction. 
%\qed \end{proof}

\subsection{Proof of Lemma \ref{prob:gftopt}}

%\begin{proof}
Let the notation $i \trades j$ denote the event that buyer $i$ trades with seller $j$ under the optimum allocation. The following is an upper bound on the optimal gain from trade.
\begin{eqnarray}
\text{OPT} &=& \sum_{i=1}^n \sum_{j=1}^m \mathbf{E}[(v_i - w_j) \mathbf{1}(i \trades j)] \notag \\
	&=& \sum_{i=1}^n \sum_{j=1}^m \mathbf{E}[v_i\mathbf{1}(i \trades j)] - \sum_{j=1}^m \sum_{i=1}^n \mathbf{E}[w_j\mathbf{1}(i \trades j)]. \label{eq:star}
\end{eqnarray}
We may rewrite the inner summation of the first term of (\ref{eq:star}) as follows.
\begin{eqnarray*}
 \sum_{j=1}^m \mathbf{E}[v_i\mathbf{1}(i \trades j)] &=&  \mathbf{E}\left[v_i\sum_{j=1}^m \mathbf{1}(i \trades j)\right] \\
	& = & \mathbf{E}\left[v_i \mathbf{1}\left(\bigvee_{j=1}^mi \trades j\right)\right] \\
	& = & \mathbf{E}\left[v_i\ \middle|\ \bigvee_{j=1}^m i \trades j\right]\mathbf{Pr}\left[\bigvee_{j=1}^m i \trades j \right].
\end{eqnarray*}
Analogously, for the inner summation of the second term of (\ref{eq:star}) we obtain
\begin{equation*}
\sum_{i=1}^n \mathbf{E}[w_j\mathbf{1}(i \trades j)] = \mathbf{E}\left[w_j\ \middle|\ \bigvee_{i=1}^n i \trades j\right]\mathbf{Pr}\left[\bigvee_{i=1}^n i \trades j \right].
\end{equation*}
Plugging these two equalities into (\ref{eq:star}), we see that 
\begin{eqnarray*}
\text{OPT} &=& \sum_{i=1}^n  \mathbf{E}\left[v_i\ \middle|\ \bigvee_{j=1}^m i \trades j\right] \mathbf{Pr}\left[\bigvee_{j=1}^m i \trades j\right] - \sum_{j=1}^m \mathbf{E}\left[w_j\ \middle|\ \bigvee_{i=1}^n i \trades j\right] \mathbf{E}\left[\bigvee_{i=1}^n i \trades j\right] \\
 & = & \sum_{i=1}^n q^B \mathbf{E}\left[v_i\ \middle|\ \bigvee_{j=1}^m i \trades j\right] - \sum_{j=1}^m q^S \mathbf{E}\left[w_j\ \middle|\ \bigvee_{i=1}^n i \trades j\right] \\
 &=& n q^B \mathbf{E}\left[v_1\ \middle|\ \bigvee_{j=1}^m 1 \trades j\right] - m q^S \mathbf{E}\left[w_1\ \middle|\ \bigvee_{i=1}^n i \trades 1\right].
\end{eqnarray*}

Now observe that the probability of the events $\left(\bigvee_{j=1}^m 1 \trades j \right)$ and $(v_1 \geq p^B)$ are both equal to $q^B$ by definition. Since the second event has the highest expected value among all events that occur with probability $q^B$, it must be that
\begin{equation*}
\mathbf{E}\left[v_1\ \middle|\ \bigvee_{j=1}^m 1 \trades j\right] \leq \mathbf{E}[v_1\ |\ v_1 \geq p^B].
\end{equation*}
Symetrically, it also holds that
\begin{equation*}
\mathbf{E}\left[w_1\ \middle|\ \bigvee_{i=1}^n i \trades 1\right] \geq \mathbf{E}[w_1\ |\ w_1 \leq p^S].
\end{equation*}
Thus, the claim follows.
%\qed\end{proof}

\subsection{Proof of Lemma \ref{lem:bound1}}

%\begin{proof}
%First, we observe that there always exist values $\overline{p}, {\overline q}^B, \overline{q}^S$ such that $\overline{p} := F^{-1}(\overline{q}^B) = G^{-1}(\overline{q}^S)$ and $n \overline{q}^B = m \overline{q}^S$.  
%To find the value $\overline{p}$, it is sufficient to raise $\overline{p}$ from $0$ till the condition $n \overline{q}^B = m \overline{q}^S$ is met with $\overline{q}^B = 1 - G(\overline{p})$ and $\overline{q}^S = F(\overline{p}).$

By Lemma \ref{lem:ordering} we know that either $p^B\geq \overline{p} \geq p^S$ or $p^S\geq \overline{p} \geq p^B$. We split the proof into the corresponding two cases.

We first look at the case where $p^B \geq \overline{p} \geq  p^S$. Define the values $\epsilon^B > 0$ and $\epsilon^S > 0$ such that $\overline{q}^B = q^B +\epsilon^B$ and $\overline{q}^S = q^S +\epsilon^S$.
%such that $n \overline{q}^B = m \overline{q}^S$ and $\overline{p}^B=\overline{p}^S =  \overline{p}$. 

Now the claim follows from the following inequalities.
\begin{eqnarray*}
& & n {\overline q}^B \mathbf{E}[v_1\ |\ v_1 \geq \overline{p}] - m \overline{q}^S \mathbf{E}[w_1\ |\ w_1 \leq \overline{p}] \\
&=& n (q^B +\epsilon^B) \left( \frac{\mathbf{Pr}[p^B \geq v_1 \geq \overline{p}]}{\mathbf{Pr}[v_1 \geq \overline{p}]}  \mathbf{E}[v_1\ |\ p^B \geq  v_1 \geq \overline{p}] + \frac{\mathbf{Pr}[v_1 \geq p^B]}{\mathbf{Pr}[v_1 \geq \overline{p}]} \mathbf{E}[v_1\ |\ v_1 \geq p^B] \right) \\  
&-& m (q^S +\epsilon^S) \left( \frac{\mathbf{Pr}[p^S \leq w_1 \leq \overline{p}]}{\mathbf{Pr}[w_1 \leq \overline{p}]}  \mathbf{E}[w_1\ |\ p^S \leq  w_1 \leq \overline{p}] + \frac{\mathbf{Pr}[w_1 \leq p^S]}{\mathbf{Pr}[w_1 \leq \overline{p}]} \mathbf{E}[w_1\ |\ w_1 \leq p^S] \right) \\
&=&  n q^B \mathbf{E}[v_1\ |\ v_1 \geq p^B]  -  m  q^S  \mathbf{E}[w_1\ |\ w_1 \leq p^S] \\
&& + n \epsilon^B  \mathbf{E}[v_1\ |\ p^B \geq  v_1 \geq \overline{p}] - m \epsilon^S  \mathbf{E}[w_1\ |\ p^S \leq  w_1 \leq \overline{p}] \\
&\geq&  n q^B \mathbf{E}[v_1\ |\ v_1 \geq p^B]  -  m  q^S  \mathbf{E}[w_1\ |\ w_1 \leq p^S], 
\end{eqnarray*}
where the second inequality follows from $\epsilon^B = \mathbf{Pr}[p^B \geq v_1 \geq \overline{p}]$ and $\epsilon^S = \mathbf{Pr}[p^S \leq w_1 \leq \overline{p}]$, and the last inequality follows from $n \epsilon^B = m \epsilon^S$ and $\mathbf{E}[v_1\ |\ p^B \geq  v_1 \geq \overline{p}] \geq \mathbf{E}[w_1\ |\ p^S \leq  w_1 \leq \overline{p}]$.

For the second case, we proceed along the same lines. It now holds that $p^S\geq \overline{p} \geq p^B$, hence we define $\epsilon^B > 0$ and $\epsilon^S > 0$ such that $q^B= \overline{q}^B +\epsilon^B$ and $q^S = \overline{q}^S +\epsilon^S$ and $n q^B = m q^S$.  
%We start from probabilities  $\overline{q}^S$ and $\overline{q}^B$  such that 
%$\overline{p}^B=\overline{p}^S =  \overline{p}$,  $n \overline{q}^B = m \overline{q}^S$,   We prove that   
%
%\[
%n q^B \expected{v_1\ |\ v_1 \geq p^B} - m q^S \expected{w_1\ |\ w_1 \leq p^S} \leq n {\overline q}^B \expected{v_1\ |\ v_1 \geq {\overline p}} - m \overline{q}^S \expected{w_1\ |\ w_1 \leq {\overline p}}
%\]
%
The claim then follows from a similar derivation. 
\begin{eqnarray*}
& & n q^B \mathbf{E}[v_1\ |\ v_1 \geq  p^B] - m q^S \mathbf{E}[w_1\ |\ w_1 \leq  p^S] \\
&=& n (\overline{q}^B +\epsilon^B) \left(\frac{\mathbf{Pr}[v_1 \geq \overline{p}]}{\mathbf{Pr}[v_1 \geq p^B]} \mathbf{E}[v_1\ |\ v_1 \geq \overline{p}] + \frac{\mathbf{Pr}[p^B \leq v_1 \leq \overline{p}]}{\mathbf{Pr}[v_1 \geq p^B]}  \mathbf{E}[v_1\ |\ p^B \leq  v_1 \leq {\overline p}]  \right) \\  
&-& m (\overline{q}^S +\epsilon^S) \left(\frac{\mathbf{Pr}[w_1 \leq \overline{p}]}{\mathbf{Pr}[w_1 \leq  p^S]} \mathbf{E}[w_1\ |\ w_1 \leq \overline{p}]  + \frac{\mathbf{Pr}[p^S \geq w_1 \geq \overline{p}]}{\mathbf{Pr}[w_1 \leq p^S]}  \mathbf{E}[w_1\ |\ p^S \geq w_1 \geq {\overline p}] \right) \\
&=&  n \overline{q}^B \mathbf{E}[v_1\ |\ v_1 \geq \overline{p}]  -  m  \overline{q}  \mathbf{E}[w_1\ |\ w_1 \leq \overline{p}]  \\
&& + n \epsilon^B  \mathbf{E}[v_1\ |\ p^B \leq  v_1 \leq \overline{p}] - m \epsilon^S  \mathbf{E}[w_1\ |\ p^S \geq  w_1 \geq \overline{p}] \\
&\leq&  n \overline{q}^B \mathbf{E}[v_1\ |\ v_1 \geq \overline{p}]  -  m  \overline{q}^S  \mathbf{E}[w_1\ |\ w_1 \leq \overline{p}], 
\end{eqnarray*}
where the last inequality follows from $n \epsilon^B = m \epsilon^S$ and $\mathbf{E}[v_1\ |\ p^B \leq  v_1 \leq \overline{p}]  \leq \mathbf{E}[w_1\ |\ p^S \geq  w_1 \geq {\overline p}].$
%\qed\end{proof}

\end{document}